\documentclass[12pt,reqno]{amsart}

\usepackage[T1]{fontenc}
\usepackage[american]{babel}
\usepackage{amssymb}
\usepackage{amsmath}
\usepackage{amscd}
\usepackage{graphicx}
\usepackage{xcolor}
\usepackage{booktabs}
\usepackage{mathtools}
\usepackage{dsfont}
\usepackage{float}
\usepackage{tikz}
\usepackage{etoolbox}
\usepackage{enumerate}
\usepackage{caption}
\usepackage{subcaption}
\usepackage[textheight=611pt,textwidth=456pt,centering]{geometry}
\usepackage[inline]{enumitem}

\usepackage[breaklinks,colorlinks,citecolor=blue,linkcolor=blue,
  pdftitle={Block-sparse recovery of semidefinite systems and
  generalized null space conditions}]{hyperref}

\graphicspath{{pics/}}

\numberwithin{equation}{section}
\newtheorem{theorem}{Theorem}
\newtheorem{proposition}[theorem]{Proposition}
\newtheorem{lemma}[theorem]{Lemma}
\newtheorem{corollary}[theorem]{Corollary}

\theoremstyle{definition}
\newtheorem{example}[theorem]{Example}
\newtheorem{remark1}[theorem]{Remark}
\newtheorem{assumption1}[theorem]{Assumption}
\newtheorem{openproblem1}[theorem]{Open problem}
\newtheorem{definition}[theorem]{Definition}

\newenvironment{remark}{\begin{remark1}\rm}{\end{remark1}}

\numberwithin{theorem}{section}

\newcounter{FNC}[page]
\def\newfootnote#1{{\addtocounter{FNC}{2}$^\fnsymbol{FNC}$%
     \let\thefootnote\relax\footnotetext{$^\fnsymbol{FNC}$#1}}}

\newcommand{\bs}{\backslash}

\newcommand{\R}{\mathds{R}}

\newcommand{\Z}{\mathds{Z}}

\newcommand{\NP}{\mathcal{N}\mathcal{P}}

\newcommand{\sym}{\mathcal{S}}
\newcommand{\psd}{\mathcal{S}_{+}}

\newcommand{\card}[1]{\lvert{#1}\rvert}
\newcommand{\abs}[1]{\lvert{#1}\rvert}
\newcommand{\norm}[1]{\lVert{#1}\rVert}
\newcommand{\define}{\coloneqq}
\newcommand{\enifed}{\eqqcolon}
\newcommand{\suchthat}{\,:\,}
\newcommand{\sprod}[2]{{#1} \bullet {#2}}
\newcommand{\T}{^\top}
\newcommand{\floor}[1]{\lfloor{#1}\rfloor}
\newcommand{\ones}{\mathds{1}}

\DeclareMathOperator{\conv}{conv}

\DeclareMathOperator{\supp}{supp}

\DeclareMathOperator{\tr}{tr}

\DeclareMathOperator{\rank}{rank}

\makeatletter
\renewcommand{\subsection}{\@startsection{subsection}{2}%
  {\z@}%
  {.7\linespacing\@plus\linespacing}%
  {.5\linespacing}%
  {\normalfont\scshape\centering}}
\makeatother

\makeatletter
\patchcmd{\@startsection}
  {\@afterindenttrue}
  {\@afterindentfalse}
  {}{}
\makeatother

\title[Block-sparse recovery of semidefinite systems and generalized
NSP\MakeLowercase{s}]{Block-sparse recovery of semidefinite systems and
  generalized null space conditions}

\author{Janin Heuer}
\address[Janin Heuer]{Technische Universit\"at Braunschweig, Institut f\"ur Analysis und Algebra, AG Algebra, Universit\"atsplatz 2, 38106 Braunschweig, Germany}
\email{janin.heuer@tu-braunschweig.de}

\author{Frederic Matter}
\address[Frederic Matter\footnote{Corresponding author.}]{Department of Mathematics, TU Darmstadt, Dolivostr.\
15, 64293 Darmstadt, Germany}
\email{matter@mathematik.tu-darmstadt.de}

\author{Marc E. Pfetsch}
\address[Marc E. Pfetsch]{Department of Mathematics, TU Darmstadt, Dolivostr.\
15, 64293 Darmstadt, Germany}
\email{pfetsch@mathematik.tu-darmstadt.de}

\author{Thorsten Theobald}
\address[Thorsten Theobald]{Goethe-Universit\"at, FB 12 -- Institut f\"ur
Mathematik, Postfach 11 19 32, 60054 Frankfurt am Main, Germany}
\email{theobald@math.uni-frankfurt.de}

\date{\today}

\thanks{The second and the third author acknowledge support by the
  EXPRESS II project within the DFG priority program CoSIP (DFG-SPP 1798).}


\begin{document}

\begin{abstract}
This article considers the recovery of low-rank matrices via a convex 
nuclear-norm minimization problem and presents two null space 
properties (NSP) which characterize uniform recovery for the 
case of block-diagonal matrices and block-diagonal positive 
semidefinite matrices.
These null-space conditions turn out to be special cases of
a new general setup, which allows to derive the mentioned NSPs
and well-known NSPs from the literature.
We discuss the relative strength of these conditions and also present a
deterministic class of matrices that satisfies the block-diagonal
semidefinite NSP.
\end{abstract}

\maketitle

\section{Introduction}

The motivating example for this article is the recovery of special
solutions of semidefinite systems. Let $A \colon \sym^n \to \R^m$ be a linear
operator from the real symmetric matrices~$\sym^n$ to $\R^m$. The goal is to
recover an unknown matrix $X^{(0)}$ from its measurements $A(X^{(0)})$ only, that
is, we are
searching for $X \in \sym^n$ with $A(X) = A(X^{(0)})$. If $A$ is
underdetermined, then it is unlikely that $X = X^{(0)}$ is recovered. In order
to increase chances of recovery, i.e.,~$X^{(0)}$ being the unique optimal solution of the corresponding recovery
problem, additional information on $X^{(0)}$ can be
used, in particular, that it has low rank. In this case,
Fazel~\cite{fazel-2002} suggested to solve the optimization
problem
\begin{equation}
  \label{eq:minrank1}
\min\, \{\rank(X) \suchthat A(X) = A(X^{(0)}),\; X \in \sym^n\}.
\end{equation}
Since the rank is a nonconvex function,~\eqref{eq:minrank1}~is hard to
solve in practice. Instead, one usually applies a convex relaxation by
replacing $\rank(X)$ by the nuclear norm $\norm{X}_{*}$, which is defined
as $\norm{X}_* = \sum_{i=1}^n \sigma_i(X)$, where $\sigma_i(X)$
are the singular values of $X$ for $i \in [n] \define \{1,\dots,n\}$.
The corresponding convex optimization problem reads
\begin{equation}
  \label{eq:relax-nuclear}
  \min\, \{\norm{X}_{*} \suchthat A(X) = A(X^{(0)}), \; X \in \sym^n\},
\end{equation}
see Recht et al.~\cite{rfp-2010} and the references therein. 

To further increase the chances of recovery, additional information might be
used, e.g., that $X^{(0)}$ is positive semidefinite or that it has a
block-diagonal structure; the corresponding optimization problems can
easily be formulated.

In fact, the mentioned settings incorporate several important special cases
that have been discussed in the compressed sensing literature. For
instance, the classical problem of recovering of a sparse vector $x$ from
measurements $A(x)$ can be obtained as a special case, 
since every vector can be interpreted as a diagonal
matrix. Since the rank of a matrix~$X$ is given by the number of nonzero
singular values of~$X$, i.e., the $\ell_0$-``norm'' of the vector of
singular values of $X$, a low-rank matrix can be considered as a natural
generalization of a sparse vector. Moreover, the nuclear norm can be seen
as a generalization of the $\ell_1$-norm, which has been used for a convex
relaxation in this case, see Chen et al.~\cite{CheDS99}. We refer to the book by Foucart
and Rauhut~\cite{foucart-rauhut-2013} for an introduction to compressed
sensing.

An important problem is to characterize when the convex
relaxation~\eqref{eq:relax-nuclear} yields $X^{(0)}$ as a unique solution
for any $X^{(0)}$ up to a given sparsity level. Such \emph{uniform recovery} can
be characterized by a so-called null space property (NSP). In the classical
case, the corresponding NSP can be found in Gribonval and Nielsen~\cite{GriN03} and
in~\cite[Theorem~4.4]{foucart-rauhut-2013}. If the vectors have to be nonnegative,
the respective NSP appears in Khajehnejad et al.~\cite{kdxh-2011} and in Zhang~\cite{Zhang2005}. For
the case of arbitrary matrices or positive semidefinite (psd) matrices,
corresponding NSPs can be found in Kong et al.~\cite{ksx-2014}, Oymak and Hassibi~\cite{oymak-hassibi-2010},
or in~\cite[Theorem~4.40]{foucart-rauhut-2013}. A very general setting for NSPs that subsumes most
of the existing NSPs has been introduced by Juditsky et
al.~\cite{JuditskyKarzanNemirovski2014}. 

To the best of the authors' knowledge, NSPs for the block-sparse
semidefinite case have not been considered in any of the existing
literature (see Section~\ref{eq:block-sparse1} for a formal definition
of block-sparsity).
Indeed, one contribution of this article is the introduction of
the corresponding NSP. It turned out that one can (with a little work)
adapt the proofs of the existing NSPs in the classical cases. A closer
investigation led us to the conclusion that all of these conditions can,
in fact, be presented and their validity proved in a very general setting,
which we describe in this article. Our setting is similar to the one of
Juditsky et al.~\cite{JuditskyKarzanNemirovski2014}, but deviates in some key points in order
to also cover nonnegative and positive semidefinite recovery.

Besides the apparent applications in compressed sensing, positive semidefinite
systems with a block-diagonal form (as formally defined in 
Definition~\ref{de:block1}) appear in various other areas. Consider a
standard semidefinite problem (SDP)
\begin{align}\label{eq:StandardSDP}
  \min\, \{A_0 \bullet X \suchthat A_p \bullet X = b_p, \; p \in
  \{1,\dots,m\},\; X \succeq 0\},
\end{align}
with $A_0,\dots,A_m \in \sym^n$, $b \in \R^m$ and
$\sprod{U}{V} \define \sum_{i,j=1}^n U_{ij}\, V_{ij}$.
To improve solving times for
such problems, sparsity in the matrices~$A_p$ should be exploited. Indeed, it is possible
to introduce a block-diagonal form on the matrix~$X$, which corresponds to the
positions of the nonzero entries in the~$A_p$ matrices. After a minor reformulation, one
ends up with an SDP of the form~\eqref{eq:StandardSDP}, where the matrices
are in block-diagonal form. Since no sparsity-related term
is added to the objective function, the optimal solution remains unchanged.
For more information on sparsity in SDPs, see,
e.g.,~\cite{FukudaKMN2001,NakataFKKM2003,VandenbergheAndersen2015}.

Moreover, block-diagonal systems appear when considering structured infeasibility in
SDPs as a generalization of the well understood structure of infeasible
linear inequality systems, see, e.g., the book~\cite{Chi2008}. Analogously
to the linear case, an irreducible infeasible subsystem (IIS) of a
semidefinite system can be defined, that is, an infeasible subsystem such that
every proper subsystem is feasible. This can be done using block-diagonal systems,
so that an IIS is given by an inclusion-minimal set of infeasible block-diagonal
subsystems. In the linear case, it is possible to fully characterize IISs by
a theorem of Gleeson and Ryan~\cite{GleR90}. For semidefinite systems
however, it turns out that this is no longer true and subsystems with
minimal block-support, i.e., \emph{block-sparse} subsystems need to be
computed in order to find an IIS, see~\cite{kpt-2019} for more details.
\smallskip

\noindent
\textbf{Main contributions.}
1. We introduce a new generalized setting for
recovery problems, which also allows to model nonnegative and positive
semidefinite constraints. Our main result is a pair of null space conditions
for the exact characterization of the uniform recovery of sufficiently 
sparse signals from their
measurements. See Theorem~\ref{thm:GeneralSettingNSP}, which
both provides a comprehensive mathematical answer on the influence of
nonnegative and semidefinite constraints and offers a versatile tool set for
a wide spectrum of reconstruction scenarios.

2. Building upon our new NSP-framework, we establish an NSP
for uniform recovery of positive semidefinite block-diagonal matrices, see
Theorem~\ref{thm:nsp_sdp}.

3. We analyze the relative strength of the new null space conditions. To this
  end, we provide a detailed classification of the most prominent subclasses of NSPs
  within the general framework, with a particular focus on the
  block-semidefinite case. This includes an NSP for the nonnegative
  block case (which had also been open so far, see Section~\ref{sec:Blocksparse}).
  In particular, we reveal the additional power of
  a nonnegative block setting for vectors and a positive semidefinite block
  setting for matrices, respectively.
  To achieve this, we construct an infinite family of instances such that
  the nonnegative block NSP captures cases which are neither captured by the
  unrestricted block NSP nor
  by the nonnegative linear NSP, see~Theorem~\ref{th:constr-family}.
 
\smallskip

The paper is structured as follows. In 
Section~\ref{sec:GeneralSettingIntroduction}, we introduce to the
setting of the paper and establish the relevant generalized NSPs.
In Section~\ref{sec:BlockPSD}, we deal with the block-diagonal semidefinite
case. Section~\ref{sec:Interrelations} then provides the interrelations
between the various NSPs.
Section~\ref{sec:Conclusion} concludes the paper with some open questions.

\smallskip

\noindent
\textbf{Notation.}
Throughout the paper, we use the following notation.
The set of symmetric real $n \times n$ matrices is denoted by $\sym^n$
and the set of positive semidefinite real $n \times n$ matrices by $\psd^n$.
Positive
semidefiniteness of a matrix $X$ is shortly written as $X \succeq 0$. The
inner matrix product of $A$, $B \in \sym^n$ is defined as
$\sprod{A}{B} \define 
\tr\big(A\T B\big) = \sum_{i,j=1}^n A_{ij}\, B_{ij}$, where $\tr(\cdot)$ is
the trace. 

On the space $\R^n$, let $\norm{x}_q \define (\sum_{i=1}^n |x_i|^q)^{1/q}$
denote the $\ell_q$-norm for some $q \ge 1$, and define
$\norm{\cdot}_0 \define \card{\{i \in [n] \suchthat x_i \neq 0\}}$ 
as the
$\ell_0$-``norm''.

\section{General sparsity structures}
\label{sec:GeneralSettingIntroduction}

As in the framework of Juditsky et al.~\cite{JuditskyKarzanNemirovski2014}, we use a
\emph{linear sensing map} $A \colon \mathcal{X} \to \R^m$ to observe
signals $x \in \mathcal{X}$ and a \emph{linear representation map}
$B \colon \mathcal{X} \to \mathcal{E}$ for mapping a signal to an
appropriate representation, where $\mathcal{X}$ and $\mathcal{E}$ are
Euclidean spaces.

By choosing $\mathcal{X}=\mathcal{E}$ and $B$ to be the
identity, this covers the classical setting of sparse recovery, where $x$
is sparse in its ``natural'' representation. However, the framework also
covers the setting where the signal~$x$ is only sparse in a suitable
representation system, with~$B$ being an appropriate transformation; this
is called the ``analysis setting'' in the compressed sensing world; for an
overview see,
e.g.,~\cite{CandesENP2011,EladMR2007,KabanavaRauhut2015,NamDEG2013}. Examples
for transformations include the discrete Fourier transform, different
wavelet
transforms~\cite{CandesDonoho2004,Groechenig2001,Mallat2008,RonShen1997,SelesnickF2009}
or a finite difference operator in total variation
minimization~\cite{CaiXu2015,ChanShen2005,NeedellWard2013}.

We introduce a set~$\mathcal{C}$ capturing further constraints
emerging from additional information like nonnegativity, and its
image~$\mathcal{D}$ under the map~$B$. Under natural assumptions, we
formulate a null space property for the uniform recovery for the set
$\mathcal{C}$ in Section~\ref{sec:GeneralSetting}.

If $\mathcal{C} = \mathcal{X}$, then our framework reduces to the framework
of~\cite{JuditskyKarzanNemirovski2014}, and our statements become the
statements concerning noise-free recovery
in~\cite{JuditskyKarzanNemirovski2014}, see
Remark~\ref{rem:DerivationJuditskySetting} below.

\begin{remark}
  Throughout the paper, we consider real vector spaces $\mathcal{X}$ and
  $\mathcal{E}$, since this is the more natural setting when considering
  nonnegative vectors. However, at least for unrestricted (block-)vectors or
  (block-diagonal)matrices, the null space properties in the subsequent
  Section~\ref{sec:GeneralSetting} also carry over without
  changes to the situation where the spaces $\mathcal{X}$ and $\mathcal{E}$ 
  are complex spaces. 
\end{remark}

\subsection{A generalized framework for sparse recovery under side constraints}
\label{sec:GeneralSetting}

Let $\mathcal{X}$, $\mathcal{E}$ be finite-dimensional Euclidean spaces and
consider an arbitrary set $\mathcal{C} \subseteq \mathcal{X}$ with~$0 \in \mathcal{C}$. Let $A \colon \mathcal{X} \to \R^m$ be a linear
sensing map and $B \colon \mathcal{X} \to \mathcal{E}$ be a linear
representation map.  Denote by
$\mathcal{D} \define \{B(x) \suchthat x \in \mathcal{C}\} \subseteq
\mathcal{E}$ the image of $\mathcal{C}$ under $B$.  Consider a norm
$\norm{\cdot}$ on $\mathcal{E}$, a set $\mathcal{P}$ of matrices
representing linear maps on $\mathcal{E}$ and a map $\nu \colon \mathcal{P}
\to \R_+$. Each map $P \in \mathcal{P}$ is
assigned a nonnegative real weight $\nu(P)$ and a linear map
$\overline{P} \colon \mathcal{E} \to \mathcal{E}$. Note that in many
examples, such as the ones in Example~\ref{ex:DeriveNonBlockSettings} 
below, $\nu(P)$ will be integer-valued, but it is not necessary to 
assume this.

\begin{remark}
  In this section, unless otherwise stated, we denote the image $F(x)$ of $x$ under a linear
  operator $F$ as $Fx$.
\end{remark}

For some real nonnegative $s$, an element $y \in \mathcal{E}$ is called
\emph{$s$-sparse} if there exists a linear map $P \in \mathcal{P}$ with
$\nu(P) \leq s$ and $Py = y$. Accordingly, an element $x \in \mathcal{X}$
is called \emph{$s$-sparse}, if its representation $Bx \in \mathcal{E}$ is
$s$-sparse. Define the set $\mathcal{P}_s \define \{P \in \mathcal{P}
\suchthat \nu(P) \leq s\}$ of linear maps that can allow
$s$-sparse elements.

For a given right-hand side $b \in \R^m$, the generalized recovery problem
now reads
\begin{align}\label{eq:GeneralSettingRecoveryProblem}
  \min\, \{ \norm{Bx}  \suchthat  Ax = b, \; x \in \mathcal{C} \}.
\end{align}
If $\mathcal{C}$ is convex,
\eqref{eq:GeneralSettingRecoveryProblem} is a convex optimization
problem. The following examples give an intuition
by showing that the setting described in this section generalizes many
important cases previously regarded in the literature.  For a (finite) set $I$, we define
the coordinate subspace $\mathcal{E}_I \define \{y \in \mathcal{E} \suchthat y_i = 0 \; \forall \, i \notin I\}$.
Additionally, we denote by $\R^I$ the space of elements with real entries
indexed by the elements of~$I$.

\begin{example}\label{ex:DeriveNonBlockSettings} \
  \begin{enumerate}[label=(\ref*{ex:DeriveNonBlockSettings}.\arabic*),wide=0pt,leftmargin=3ex]
  \item\label{ex:LinearSettingDerivation}
    \emph{Recovery of sparse vectors by $\ell_1$-minimization}\\
    For the recovery of sparse vectors $x \in \R^n$, let
    $\mathcal{X} = \mathcal{E} = \mathcal{C} = \R^n$, $B$ be
    the identity and $\norm{\cdot} = \norm{\cdot}_1$; 
    then $\mathcal{D} = \R^n$. Let $\mathcal{P}$ be the set of orthogonal
    projectors onto all coordinate subspaces of $\R^n$, and define
    $\overline{P} \define I_n - P$, where $I_n$ denotes the identity mapping on
    $\R^n$. If $S$ is the index set of components on which~$P$ projects, then
    $PBx = Px = x_S$, where $x_S \in \R^n$ equals $x$ on $S$ and is 0 otherwise.
    Similarly, for the complement $\overline{S} \define [n] \setminus S$,
    $\overline{P}Bx = x_{\overline{S}}$.
    Define the nonnegative weight $\nu(P) \define \rank(P)$, so
    that $\nu(P)$ is the number of nonzero components of the subspace~$P$
    projects onto. The notion of general sparsity reduces to
    the classical sparsity of nonzero entries in a vector $x \in \R^n$ (if
    $Px = x$, then $\norm{x}_0 \leq \nu(P)$), and
    the recovery problem~\eqref{eq:GeneralSettingRecoveryProblem} becomes
    classical $\ell_1$-minimization.
  \item\label{ex:LinearNonnegSettingDerivation}
    \emph{Recovery of sparse nonnegative vectors by $\ell_1$-minimization}\\
    For the recovery of nonnegative vectors let $\mathcal{X}$,
    $\mathcal{E}$, $B$, $\mathcal{P}$, $\nu(P)$, $\overline{P}$,
    $\norm{\cdot}$ be defined as in the previous example, and let
    $\mathcal{C} = \R^n_+$, implying $\mathcal{D} = \R^n_+$.   As before, the notion of general
    sparsity now simplifies to the classical sparsity of nonzero entries in
    a nonnegative vector $x \in \R^n_+$, and the recovery
    problem~\eqref{eq:GeneralSettingRecoveryProblem} becomes nonnegative
    $\ell_1$-minimization with $PBx = x_S$ and $\overline{P}Bx = x_{\overline{S}}$.
  \item \label{ex:MatrixSettingDerivation}
    \emph{Recovery of low-rank matrices by nuclear norm minimization}\\
    Let
    $\mathcal{X} = \mathcal{E} = \mathcal{C} = \R^{n_1\times
      n_2}$. Let the representation map $B$ be the identity (thus, $\mathcal{D} = \R^{n_1\times n_2}$), and let the norm
    $\norm{\cdot}$ be the nuclear norm $\norm{\cdot}_*$.
    For some
    positive integer~$k$ and a set $I \subseteq [k]$, define
    the matrix $T_I^k \in \R^{k \times k}$ to be a matrix with ones on the
    diagonal at positions $(i,i)$ for $i \in I$ and zeros elsewhere. Let~$\mathcal{O}^k$
    be the set of $k \times k$ orthogonal matrices. Then define the
    set $\mathcal{P}$ of projections
    $P \colon \R^{n_1\times n_2} \to \R^{n_1 \times n_2}$ as
    \begin{align*}
      \hspace*{\leftmargin}
      \mathcal{P} \define \big\{ X \mapsto U\,T^{n_1}_I\,U\T\,X\,V\,T^{n_2}_I\,V\T \suchthat
      U \in \mathcal{O}^{n_1},\, V \in \mathcal{O}^{n_2},\, I \subseteq \{1, \dots, \min\{n_1,n_2\}\}\big\}.
    \end{align*}
    For $P \in \mathcal{P}$ defined by $U \in \mathcal{O}^{n_1}$, $V \in \mathcal{O}^{n_2}$ and index
    set $I$, define the nonnegative weight $\nu(P)$ as $\nu(P) = \card{I}$,
    and $\overline{P}$ as
    \begin{align*}
      \hspace*{\leftmargin}
      X \mapsto U\,(I_{n_1} - T^{n_1}_{I})\,U\T\,X\,V\,(I_{n_2} - T^{n_2}_I)\,V\T,
    \end{align*}
    where $I_{n_i}$ denotes the identity matrix of size $n_i$, so that
    $I_{n_i} - T^{n_i}_{I} = T^{n_i}_{[n_i]\bs I}$.\footnote{Note that
      $U\,T^{n_1}_I\,U\T\,X\,V\,T^{n_2}_I\,V\T$ and
      $U\,(I_{n_1} - T^{n_1}_{I})\,U\T\,X\,V\,(I_{n_2} - T^{n_2}_I)\,V\T$ denote matrix
      products, since $U$, $T^{n_1}_I$, $T^{n_2}_I$ $X$, and $V$ are matrices, and
      not linear maps.}
    
    The intuition behind these projections is as follows. If $U$, $V$ are
    chosen such that $X = U\Sigma V\T$ is the singular value decomposition
    of $X$, then $P$ first projects~$X$ onto $\Sigma$
    containing the singular values
    $\sigma_1(X) \geq \cdots \geq \sigma_{\min\{n_1,n_2\}}(X)$, then sets
    $\sigma_i(X) = 0$ for all $i \notin I$ via left- and
    right-multiplication of $T^{n_1}_I$ and $T^{n_2}_I$, respectively, and transforms the resulting
    diagonal matrix~$\tilde{\Sigma}$ back by $U\, \tilde{\Sigma}\, V\T$.

    A matrix $X \in \R^{n_1\times n_2}$ is rank-$s$-sparse, i.e., there exist
    at most $s$ nonzero singular values, if and only if there exists a
    projection $P \in \mathcal{P}$ with corresponding index set~$I$ with
    $PX = X$ and $\card{I} \leq s$. Thus, $\sigma_i(X) = 0$ for all
    $i \notin I$, and accordingly, $\rank(X) \leq s$. Therefore, the
    recovery problem~\eqref{eq:GeneralSettingRecoveryProblem} becomes
    low-rank matrix recovery, and sparsity translates to low-rankness.
  \item\label{ex:MatrixPSDSettingDerivation}
    \emph{Recovery of positive definite low-rank matrices by nuclear norm minimization}\\
    For the recovery of positive semidefinite matrices, let $\psd^n$ be the set of psd matrices of size $n \times n$.
    Then consider $\mathcal{X} = \mathcal{E} = \sym^n$,
    $\mathcal{C} = \sym^n_+$, and let $B$ be the identity
    map (thus, $\mathcal{D} = \sym^n_+$). The definitions of $\mathcal{P}$, $\overline{P}$, $\nu(P)$ and
    $\norm{\cdot}$ are as in the previous example. Again, the notion
    of sparsity simplifies to low-rankness. Recovery
    problem~\eqref{eq:GeneralSettingRecoveryProblem} becomes low-rank
    recovery for positive semidefinite matrices.
  \end{enumerate}
\end{example}

Examples with a nontrivial representation map
$B \colon \mathcal{X} \to \mathcal{E}$ are given by settings in which the
vectors or matrices obey a certain block-structure or block-diagonal form, 
even with overlapping
blocks. These settings are not as well studied as the settings described
above, so that they are discussed in the subsequent sections in more
detail.

In order to characterize the cases in which a sparse element $x^{(0)}$ can
be recovered from its measurements $b = Ax^{(0)}$
using the recovery problem~\eqref{eq:GeneralSettingRecoveryProblem}, we consider the following
assumptions on the sets~$\mathcal{C}$, $\mathcal{D}$, $\mathcal{P}$ and the
norm~$\norm{\cdot}$.

\begin{enumerate}[label=(A\arabic*)]
\item \label{eq:GeneralSettingAssump1} For every $P \in \mathcal{P}$ it
  holds that
  \begin{itemize}[label=$\circ$,leftmargin=3ex]
  \item $P^2 = P$, i.e., $P$ is a projector, and
  \item $Py \in \mathcal{D}$ for all $y \in \mathcal{D}$.
  \end{itemize}
  Moreover, $B \colon \mathcal{X} \to \mathcal{E}$ is injective, and for
  all $c_1$, $c_2 \in \mathcal{C}$, $c_1 + c_2 \in \mathcal{C}$ holds.
\item \label{eq:GeneralSettingAssump2} For every $P \in \mathcal{P}$ the
  corresponding linear map
  $\overline{P} \colon \mathcal{E} \to \mathcal{E}$ satisfies
  \begin{itemize}[label=$\circ$,leftmargin=3ex]
  \item $\overline{P}P = 0$, and
  \item $\overline{P}y \in \mathcal{D}$ for all $y \in \mathcal{D}$.
  \end{itemize}
\item \label{eq:GeneralSettingAssump3} For all $y \in \mathcal{E}$ and all
  $P \in \mathcal{P}$ it holds that $y = Py + \overline{P}y$.
\end{enumerate}

\begin{enumerate}[label=(A4\alph*)]
\item \label{eq:GeneralSettingAssump4a} For all $s \geq 0$,
  $P \in \mathcal{P}_s$, for all $x$, $z \in \mathcal{C}$ with $PBx = Bx$
  (i.e., $x$ being $s$-sparse) and $v\define x-z$ and all
  $v^{(1)},\, v^{(2)} \in \mathcal{C}$ with $v = v^{(1)}-v^{(2)}$ it holds
  that
  \begin{align*}
    \norm{Bx} \leq \norm{Bz} + \norm{PBv^{(1)}} -
    \norm{PBv^{(2)}} - \norm{\overline{P}Bv}.
  \end{align*}
\item \label{eq:GeneralSettingAssump4b} For all $s \geq 0$,
  $P \in \mathcal{P}_s$, for all $x$, $z \in \mathcal{C}$ with $PBx = Bx$
  (i.e., $x$ being $s$-sparse) and $v\define x-z$ there exist
  $v^{(1)},\, v^{(2)} \in \mathcal{C}$ with $v = v^{(1)}-v^{(2)}$ and
  \begin{align*}
    \norm{Bx} \leq \norm{Bz} + \norm{PBv^{(1)}} -
    \norm{PBv^{(2)}} - \norm{\overline{P}Bv}.
  \end{align*}
\end{enumerate}

Note that
Assumptions~\ref{eq:GeneralSettingAssump1}--\ref{eq:GeneralSettingAssump3}
are satisfied in the different settings in
Example~\ref{ex:DeriveNonBlockSettings}, because $\mathcal{P}$ consists of
orthogonal projections. Only Assumptions~\ref{eq:GeneralSettingAssump4a}
and~\ref{eq:GeneralSettingAssump4b} remain to be verified. A discussion of
these assumptions in the settings of
Example~\ref{ex:DeriveNonBlockSettings} will follow after the main result
of this section.

We can now define two versions of a null space property.

\begin{definition}\label{def:GeneralSettingNSPa}
  The linear sensing map $A$ satisfies the \emph{general null space
    property of type~I} of order $s$ for the set $\mathcal{C}$ if and only
  if for all $v \in (\mathcal{N}(A) \cap (\mathcal{C}+(-\mathcal{C})))$
  with $Bv \neq 0$ and all $P \in \mathcal{P}_s$ it holds that
  \begin{equation}\label{eq:defGeneralSettingNSPa}
    \begin{aligned}
      -\overline{P}Bv \in \mathcal{D} \; \implies \; \exists \; v^{(1)},\,
      v^{(2)} &\in \mathcal{C} \text{ with }
      v = v^{(1)} - v^{(2)} \text{ and } \\
      &\norm{PBv^{(1)}} - \norm{PBv^{(2)}} < \norm{\overline{P}Bv},
    \end{aligned}\tag{$\text{NSP-I}^\mathcal{C}$}
  \end{equation}
  where~$\mathcal{N}(A) \define \{v \in \mathcal{X} \suchthat Av = 0 \}$ is
  the null space of the linear sensing map~$A$.
\end{definition}

\begin{definition}\label{def:GeneralSettingNSPb}
  The linear sensing map $A$ satisfies the \emph{general null space
    property of type~II} of order $s$ for the set $\mathcal{C}$ if and only
  if for all $v \in (\mathcal{N}(A) \cap (\mathcal{C}+(-\mathcal{C})))$
  with $Bv \neq 0$ and all $P \in \mathcal{P}_s$ it holds that
  \begin{equation}\label{eq:defGeneralSettingNSPb}
    \begin{aligned}
      -\overline{P}Bv \in \mathcal{D} \; \implies \; \forall \; v^{(1)},\,
      v^{(2)} &\in \mathcal{C} \text{ with } v = v^{(1)} - v^{(2)}:\\
      &\norm{PBv^{(1)}} - \norm{PBv^{(2)}} < \norm{\overline{P}Bv}.
    \end{aligned}\tag{$\text{NSP-II}^\mathcal{C}$}
  \end{equation}
\end{definition}
\noindent

\begin{remark}
  We emphasize that~\eqref{eq:defGeneralSettingNSPa}
  and~\eqref{eq:defGeneralSettingNSPb} only differ in the quantifiers
  (namely, $\exists \; v^{(1)},\, v^{(2)}$ versus $\forall \; v^{(1)},\, v^{(2)}$) , 
  and these are closely connected to the quantifiers
  in~\ref{eq:GeneralSettingAssump4a} and~\ref{eq:GeneralSettingAssump4b}:
  If
  \begin{align}\label{eq:ConditionA4ab}
    \norm{Bx} \leq \norm{Bz} + \norm{PBv^{(1)}} - \norm{PBv^{(2)}} -
    \norm{\overline{P}Bv}
  \end{align}
  is satisfied \emph{for all} $v^{(1)},\, v^{(2)} \in \mathcal{C}$ with
  $v = v^{(1)}-v^{(2)}$, there only need to \emph{exist}
  $v^{(1)},\, v^{(2)} \in \mathcal{C}$ with $v = v^{(1)}-v^{(2)}$ such that
  \begin{align}\label{eq:ConditionNSPab}
    \norm{PBv^{(1)}} - \norm{PBv^{(2)}} - \norm{\overline{P}Bv} < 0
  \end{align}
  holds (if $-\overline{P}Bv \in \mathcal{D}$), which
  is~\eqref{eq:defGeneralSettingNSPa}. Otherwise, if there only
  \emph{exist} $v^{(1)},\, v^{(2)} \in \mathcal{C}$ with
  $v = v^{(1)}-v^{(2)}$ such that~\eqref{eq:ConditionA4ab} holds,
  then~\eqref{eq:ConditionNSPab} must be satisfied by \emph{all}
  $v^{(1)},\, v^{(2)} \in \mathcal{C}$ with $v = v^{(1)}-v^{(2)}$, which
  is~\eqref{eq:defGeneralSettingNSPb}. As we will see in the subsequent
  theorem, under Assumption~\ref{eq:GeneralSettingAssump4a},
  \eqref{eq:defGeneralSettingNSPa} characterizes uniform recovery, and if
  Assumption~\ref{eq:GeneralSettingAssump4b} is satisfied,
  then~\eqref{eq:defGeneralSettingNSPb} is needed for uniform recovery.

  The reason for formulating these two slightly different null space
  properties is that in the settings described in
  Examples~\ref{ex:DeriveNonBlockSettings}, the unrestricted cases, i.e.,
  $\mathcal{C} = \mathcal{X}$, satisfy
  Assumption~\ref{eq:GeneralSettingAssump4b} and thus
  need~\eqref{def:GeneralSettingNSPb} for uniform recovery, whereas in the
  restricted cases, Assumption~\ref{eq:GeneralSettingAssump4a} holds, so
  that~\eqref{eq:defGeneralSettingNSPa} suffices for uniform recovery.
\end{remark}

The following main result of this section states that the above defined null space
properties~\eqref{eq:defGeneralSettingNSPb}
and~\eqref{eq:defGeneralSettingNSPa} exactly characterize the uniform
recovery of a sufficiently sparse~$x \in \mathcal{C}$ from its
measurements~$b = Ax$ using~\eqref{eq:GeneralSettingRecoveryProblem},
depending on which assumptions are satisfied.

\begin{theorem}\label{thm:GeneralSettingNSP}
  Suppose that
  Assumptions~\ref{eq:GeneralSettingAssump1},~\ref{eq:GeneralSettingAssump2}
  and~\ref{eq:GeneralSettingAssump3} are satisfied. Let $A$ be a linear
  sensing map and $s \geq 1$.
  \begin{enumerate}[leftmargin=4ex]
  \item If Assumption~\ref{eq:GeneralSettingAssump4a} is satisfied, then
    the following statements are equivalent:
    \begin{itemize}
    \item[(i)] Every $s$-sparse $x^{(0)} \in \mathcal{C}$ is the unique
      solution of~\eqref{eq:GeneralSettingRecoveryProblem} with
      $b = Ax^{(0)}$.
    \item[(ii)] $A$ satisfies the general null space
      property~\eqref{eq:defGeneralSettingNSPa} of order $s$ for the set
      $\mathcal{C}$.
    \end{itemize}
    
  \item If Assumption~\ref{eq:GeneralSettingAssump4b} is satisfied, then
    the following statements are equivalent:
    \begin{itemize}
    \item[(i)] Every $s$-sparse $x^{(0)} \in \mathcal{C}$ is the unique
      solution of~\eqref{eq:GeneralSettingRecoveryProblem} with
      $b = Ax^{(0)}$.
    \item[(ii)] $A$ satisfies the general null space
      property~\eqref{eq:defGeneralSettingNSPb} of order $s$ for the set
      $\mathcal{C}$.
    \end{itemize}
  \end{enumerate}
\end{theorem}

\begin{proof}
  For the first equivalence, let $s \geq 1$ and suppose
  Assumptions~\ref{eq:GeneralSettingAssump1}--\ref{eq:GeneralSettingAssump3}
  and~\ref{eq:GeneralSettingAssump4a} are satisfied.

  Assume that if $Ax = b$ has an $s$-sparse
  solution~$x^{(0)} \in \mathcal{C}$, then $x^{(0)}$ is the unique solution
  of~\eqref{eq:GeneralSettingRecoveryProblem}. Let $P \in \mathcal{P}_s$
  and $v \in (\mathcal{N}(A) \cap (\mathcal{C}+(-\mathcal{C})))$ with
  $Bv \neq 0$ and $-\overline{P}Bv \in \mathcal{D}$.  Since~$v \in \mathcal{C}+(-\mathcal{C})$, there exist $v^{(1)}$,
  $v^{(2)} \in \mathcal{C}$ with $v = v^{(1)} - v^{(2)}$. Define
  \begin{align*}
    w_s^{(1)} \define PBv^{(1)},\quad w_s^{(2)} \define PBv^{(2)},\quad w_{\overline{s}}
    \define \overline{P}Bv.
  \end{align*}
  Since by Assumption~\ref{eq:GeneralSettingAssump1}, $PBx \in \mathcal{D}$
  for all $x \in \mathcal{C}$ , there exist $v_s^{(1)}$, $v_s^{(2)}$,
  $v_{\overline{s}} \in \mathcal{C}$ with $Bv_s^{(1)} = w_s^{(1)}$,
  $Bv_s^{(2)} = w_s^{(2)}$ and $Bv_{\overline{s}} = w_{\overline{s}}$. Due to
  Assumption~\ref{eq:GeneralSettingAssump3}
  \begin{align*}
    Bv = PBv + \overline{P}Bv = PBv^{(1)} - PBv^{(2)} + \overline{P}Bv
    &= Bv_s^{(1)} - Bv_s^{(2)} + Bv_{\overline{s}}\\
    &= B(v_s^{(1)} - v_s^{(2)} + v_{\overline{s}}),
  \end{align*}
  which implies $v = v_s^{(1)} - v_s^{(2)} + v_{\overline{s}}$, since $B$
  is injective by Assumption~\ref{eq:GeneralSettingAssump1}. Accordingly,
  \begin{align*}
    0 = Av = A(v_s^{(1)} - v_s^{(2)} + v_{\overline{s}}) \quad
    \Leftrightarrow \quad A(v_s^{(2)} - v_{\overline{s}}) = Av_s^{(1)}.
  \end{align*}
  By Assumption~\ref{eq:GeneralSettingAssump1},
  $PBv_s^{(1)} = Pw_s^{(1)} = PPBv^{(1)} = PBv^{(1)} = w_s^{(1)} =
  Bv_s^{(1)}$, i.e., $v_s^{(1)}$ is $s$-sparse. Moreover, by
  Assumption~\ref{eq:GeneralSettingAssump1}, $Py \in \mathcal{D}$ for all
  $y \in \mathcal{D}$, $B$ is injective and $c_1 + c_2 \in \mathcal{C}$ for
  all $c_1$, $c_2 \in \mathcal{C}$. Thus, $v^{(1)} \in \mathcal{C}$ implies
  $Bv_s^{(1)} = w_s^{(1)} = PBv^{(1)} \in \mathcal{D}$. Since
  $\mathcal{D} = \{Bx \suchthat x \in \mathcal{C}\}$, there exists
  $u \in \mathcal{C}$ with $Bu = Bv_s^{(1)} \in \mathcal{D}$. This in
  turn implies $v_s^{(1)} = u \in \mathcal{C}$, since $B$ is
  injective. Analogously, $v_s^{(2)} \in \mathcal{C}$ and
  $-v_{\overline{s}} \in \mathcal{C}$, since $v^{(2)} \in \mathcal{C}$ and
  $-\overline{P}Bv \in \mathcal{D}$ per assumption. Altogether, this yields
  $v_s^{(2)} - v_{\overline{s}} \in \mathcal{C}$, since
  $c_1 + c_2 \in \mathcal{C}$ for all $c_1$, $c_2 \in \mathcal{C}$. The
  uniqueness property of $A$ for the $s$-sparse $v_s^{(1)}$ now implies
  \begin{align*}
    & \norm{Bv_s^{(1)}} < \norm{Bv_s^{(2)} -
     B v_{\overline{s}}} \leq \norm{Bv_s^{(2)}} +
      \norm{Bv_{\overline{s}}} \\
    \Leftrightarrow \quad
    & \norm{Bv_s^{(1)}} - \norm{Bv_s^{(2)}} -
      \norm{Bv_{\overline{s}}} < 0 \\
    \Leftrightarrow \quad
    & \norm{PBv^{(1)}} - \norm{PBv^{(2)}} -
      \norm{\overline{P}Bv} < 0.               
  \end{align*}

  For the reverse direction, assume $A$ satisfies the general null space
  property~\eqref{eq:defGeneralSettingNSPa} of order $s$ for the set
  $\mathcal{C}$. Let $x$, $z \in \mathcal{C}$ with $Bx \neq Bz$, $Ax = Az$
  and $x$ being $s$-sparse, i.e., there exists $P \in \mathcal{P}_s$ with
  $PBx = Bx$. Define
  $v\define x-z \in \mathcal{N}(A)\cap (\mathcal{C}+(-\mathcal{C}))$ with
  $-\overline{P}Bv = -\overline{P}Bx + \overline{P}Bz = -\overline{P}PBx +
  \overline{P}Bz = \overline{P}Bz \in \mathcal{D}$, since
  $\overline{P}P = 0$ and $\overline{P}y \in \mathcal{D}$ for all
  $y \in \mathcal{D}$ (Assumption~\ref{eq:GeneralSettingAssump2}). The
  general null space property~\eqref{eq:defGeneralSettingNSPa} implies the
  existence of $v^{(1)}$, $v^{(2)} \in \mathcal{C}$ with
  $v = v^{(1)} - v^{(2)}$ and
  \begin{align} \label{eq:ConditionGeneralNSP}
    \norm{PBv^{(1)}} - \norm{PBv^{(2)}} -
    \norm{\overline{P}Bv} < 0.
  \end{align}
  Together with Assumption~\ref{eq:GeneralSettingAssump4a} this yields
  \begin{align*}
    \norm{Bx} \leq \norm{Bz} +
    \norm{PBv^{(1)}} - \norm{PBv^{(2)}} -
    \norm{\overline{P}Bv} < \norm{Bz}.
  \end{align*}
  This shows that $x$ must be the unique solution
  of~\eqref{eq:GeneralSettingRecoveryProblem}, which completes the proof
  of the first equivalence.
  \smallskip

  For the second equivalence, note that the forward direction implying the
  general null space property~\eqref{eq:defGeneralSettingNSPb} is
  completely analogous to above. The proof for the reverse direction needs
  one small adjustment. In this case, the general null space
  property~\eqref{eq:defGeneralSettingNSPb} implies
  that~\eqref{eq:ConditionGeneralNSP} holds for all $v^{(1)}$,
  $v^{(2)} \in \mathcal{C}$ with $v = v^{(1)} - v^{(2)}$. Since
  $v \in \mathcal{C}+(-\mathcal{C})$, there exists at least one
  decomposition $v = v^{(1)} - v^{(2)}$ with $v^{(1)}$,
  $v^{(2)} \in \mathcal{C}$. Together with
  Assumption~\ref{eq:GeneralSettingAssump4b} this implies
  \begin{align*}
    \norm{Bx} \leq \norm{Bz} +
    \norm{PBv^{(1)}} - \norm{PBv^{(2)}} -
    \norm{\overline{P}Bv} < \norm{Bz},
  \end{align*}
  which concludes the proof of the second equivalence.
\end{proof}

\begin{remark}\label{rem:DerivationJuditskySetting}
  Let $\mathcal{C} = \mathcal{X}$ and $\mathcal{D} = \mathcal{E}$. Then our
  setting simplifies to the framework
  in~\cite{JuditskyKarzanNemirovski2014}. Clearly, under Assumptions~A.1--A.3 in~\cite{JuditskyKarzanNemirovski2014},
  Assumptions~\ref{eq:GeneralSettingAssump1},~\ref{eq:GeneralSettingAssump2}
  and~\ref{eq:GeneralSettingAssump4b} are
  satisfied. In this case,~\eqref{eq:defGeneralSettingNSPb} is only a
  sufficient condition, and it implies the
  sufficient condition in~\cite[Lemma~3.1]{JuditskyKarzanNemirovski2014}, namely
  \begin{align}
    \label{eq:NemirovskiNSP}
    \norm{PBv} < \norm{\overline{P}Bv}
  \end{align}
  for all $P \in \mathcal{P}_s$ and all $v \in \mathcal{N}(A)$,
  $Bv \neq 0$. If additionally Assumption~\ref{eq:GeneralSettingAssump3}
  holds, then~\eqref{eq:NemirovskiNSP} and~\eqref{eq:defGeneralSettingNSPb}
  are also necessary conditions and in fact equivalent.
\end{remark}

For all the settings derived in Example~\ref{ex:DeriveNonBlockSettings},
specific NSPs are already known in the literature. In the next example, we
demonstrate how these NSPs emerge from the generalized null space
properties~\eqref{eq:defGeneralSettingNSPa}
and~\eqref{eq:defGeneralSettingNSPb}.
As already mentioned before,
Assumptions~\ref{eq:GeneralSettingAssump1}--\ref{eq:GeneralSettingAssump3}
are satisfied in all four settings of
Example~\ref{ex:DeriveNonBlockSettings}. In case that the NSPs are satisfied, the null space characterizations provide 
algorithmically tractable algorithms to find the solution, using linear
programming, semidefinite programming or the convex optimization problem
of minimizing the nuclear norm, respectively. However, already in the
special case of recovery of sparse vectors, it is $\NP$-hard to check
whether a given sensing matrix~$A$ satisfies the classical null space
property~\cite{PfetschTillmann2014}.

~\
\begin{example}\label{ex:CompareNonBlockSettings} \
  \begin{enumerate}[label=(\ref*{ex:CompareNonBlockSettings}.\arabic*),wide=0pt,leftmargin=3ex]
  \item\label{ex:LinearSettingNSP} \emph{Recovery of sparse vectors by
      $\ell_1$-minimization,
      Example~\ref{ex:LinearSettingDerivation} continued}\\
    The following example shows that
    Assumption~\ref{eq:GeneralSettingAssump4a} is violated. Let
    $z = (2,0,0)\T$, $x = (0,-1,0)\T$ so that $v = x - z = (-2,-1,0)\T$. Let
    $P$ be the projection onto the first two coordinates. The decomposition
    $v^{(1)} = (8,9,0)\T$ and $v^{(2)} = (10,10,0)\T$ yields
    \begin{align*}
      \hspace*{\leftmargin}
      \norm{Bx}_1 = 1 > 2 + 17 - 20 - 0 = \norm{Bz}_1 + \norm{PBv^{(1)}}_1
      - \norm{PBv^{(2)}}_1 - \norm{\overline{P}Bv}_1.
    \end{align*}
    However, Assumption~\ref{eq:GeneralSettingAssump4b} is satisfied, so
    that~\eqref{eq:defGeneralSettingNSPb} characterizes uniform
    recovery. For the decomposition $v = v-0$, where $0$ denotes the
    all-zero vector, condition~\eqref{eq:defGeneralSettingNSPb} simplifies
    to the regular null space property (see, e.g.,
    \cite{foucart-rauhut-2013}):
    \begin{align}\label{eq:linearNSP}
      \hspace*{\leftmargin}
      \norm{v_S}_1 < \norm{v_{\overline{S}}}_1 \quad \forall \, v\in \mathcal{N}(A)
        \bs \{0\}, \; \forall \, S \subseteq [n],\, \card{S} \leq s, \tag{NSP}
    \end{align}
    where $S$ denotes the index set of components on which $P$ projects, and $\overline{S} \define [n] \bs S$.
    It can be shown that Condition~\eqref{eq:defGeneralSettingNSPb} and
    $\norm{v_S}_1 < \norm{v_{\overline{S}}}_1$ are equivalent.

  \item\label{ex:LinearNonnegSettingNSP} \emph{Recovery of sparse
      nonnegative vectors by $\ell_1$-minimization,
      Example~\ref{ex:LinearNonnegSettingDerivation} continued}\\
    In contrast to the previous setting,
    Assumption~\ref{eq:GeneralSettingAssump4a} is satisfied for the
    recovery of sparse nonnegative vectors. It can be shown that the
    general null space property~\eqref{eq:defGeneralSettingNSPb} of order
    $s$ for the set $\mathcal{C}$ is equivalent to the known nonnegative
    null space property~\cite{kdxh-2011,Zhang2005}:
    \begin{align}\label{eq:linearnonnegNSP}
      \hspace*{\leftmargin}
      v_{\overline{S}} \leq 0 \, \implies \, \sum_{i \in S} v_i <
      \norm{v_{\overline{S}}}_1,  \; \forall \, v \in \mathcal{N}(A)\bs
      \{0\}, \; \forall \, S \subseteq [n],\, \card{S} \leq s, \tag{$\text{NSP}_{\geq 0}$}
    \end{align}
    where again $S$ denotes the index set of components on which $P$ projects.

  \item \label{ex:MatrixSettingNSP} \emph{Recovery of low-rank matrices by nuclear norm
      minimization, Example~\ref{ex:MatrixSettingDerivation} continued}\\
    Since vectors can be interpreted as diagonal matrices, the same
    counterexample as Example~\ref{ex:LinearSettingNSP} shows that
    Assumption~\ref{eq:GeneralSettingAssump4a} is not fulfilled. However,
    it can be shown that Assumption~\ref{eq:GeneralSettingAssump4b} is
    satisfied, so that the general null space
    property~\eqref{eq:defGeneralSettingNSPb} characterizes uniform
    recovery for low-rank matrix matrices. Using the decomposition
    $V = V - 0$, condition~\eqref{eq:defGeneralSettingNSPb} simplifies
    to the well-known null space
    property~\cite{oymak-hassibi-2010,RechtXuHassibi2008},
    \cite[Theorem~4.40]{foucart-rauhut-2013}:
    \begin{align}\label{eq:matrixNSP}
      \hspace*{\leftmargin}
      \sum_{j\in S} \sigma_j(V) < \sum_{j \in \overline{S}} \sigma_j(V),
      \; \forall \, V \in \mathcal{N}(A) \bs \{0\}, \; \forall \, S
      \subseteq [\min\{n_1,n_2\}],\, \card{S} \leq s,  \tag{$\text{NSP}^*$}
    \end{align}
    where $\sigma(V)$ is the vector of singular values of $V$, and~$S$ is
    connected to a projection $P \in \mathcal{P}$ with index set~$I$ by $S = I$. For
    symmetric matrices~$X \in \sym^n$ this simplifies to
    \begin{align*}
      \hspace*{\leftmargin}
      \norm{\lambda_S(V)}_1 < \norm{\lambda_{\overline{S}}(V)}_1,
      \; \forall \, V \in \mathcal{N}(A) \bs \{0\}, \; \forall \, S
      \subseteq [n],\, \card{S} \leq s,
    \end{align*}
    where $\lambda(V)$ is the vector of eigenvalues of~$V$.
    
  \item\label{ex:MatrixPSDSettingNSPS} \emph{Recovery of positive
      semidefinite low-rank matrices by nuclear
      norm minimization, Example~\ref{ex:MatrixPSDSettingDerivation} continued}\\
    Again, in contrast to the previous setting,
    Assumption~\ref{eq:GeneralSettingAssump4a} is satisfied for recovery of
    positive semidefinite low-rank matrices. The general null space
    property~\eqref{eq:defGeneralSettingNSPa} simplifies to the following
    null space property~\cite{ksx-2014,oymak-hassibi-2010}:
    \begin{equation}\label{eq:matrixpsdNSP}
      \hspace*{\leftmargin}
    \begin{aligned}
      &\lambda_{\overline{S}}(V) \leq 0 \,
      \implies \, \sum_{j \in S} \lambda_j(V) < \norm{\lambda_{\overline{S}}(V)}_1,
      \\ &\quad \forall \, V \in (\mathcal{N}(A) \cap
      \sym^n) \bs \{0\}, \; \forall \, S \subseteq [n],\, \card{S} \leq s,
    \end{aligned}\tag{$\text{NSP}_{\succeq 0}^*$}
  \end{equation}
  where $\lambda(V)$ is the vector of eigenvalues of $V$.
  \end{enumerate}
\end{example}

\begin{remark}\label{rem:NonnegNSPvsNSP}
  The formulation of the nonnegative null space
  property~\eqref{eq:linearnonnegNSP} in
  Example~\ref{ex:LinearNonnegSettingNSP} already indicates
  that~\eqref{eq:linearnonnegNSP} is weaker than~\eqref{eq:linearNSP},
  since for the left hand side $\sum_{j \in S} v_j \leq \norm{v_s}_1$
  holds, and additionally, if the condition $v_{\overline{S}} \leq 0$ for
  all $v \in \mathcal{N}(A) \bs \{0\}$ is violated, then the inequality
  $\sum_{i \in S} v_i < \norm{v_{\overline{S}}}_1$ need not hold,
  see Example~\ref{ex:block_nsp_psd_vs_block_nsp} for an explicit
  case.
\end{remark}

\begin{remark}
From the viewpoint of an \emph{ordered} vector space, the condition
in~\eqref{eq:defGeneralSettingNSPa} can be interpreted as follows:
Let $(V, \le)$ be a finite-dimensional ordered real vector space,
i.e., a finite-dimensional real vector space $V$ with a partial
order $\le$. The \emph{positive cone}
\[
  C_V \define \{ x \in V \suchthat x \ge 0 \}
\]
is a convex cone with $C_V \cap (-C_V) = \{0\}$. If $C_V$ is
full-dimensional (which is the case, for instance for $\R^n$
with the usual ordering on vectors, or for the space of
symmetric real~$n \times n$-matrices with the usual L\"owner partial order:
$A \preceq B :\iff B-A \succeq 0$), we have
$C_V - C_V = V$ due to the following lemma.

\begin{lemma}
  Let $K \subseteq \R^n$ be a convex cone. Then $K-K = \R^n$
  if and only if $K$ is full-dimensional.
\end{lemma}

\begin{proof}
   Clearly, if $K$ is not full-dimensional, then
   $K-K$ is not full-dimensional. For the converse direction see, e.g.,
   Ahmadi and Hall~\cite[Lemma 1]{AhmH18}.
\end{proof}

Thus, if $\mathcal{C} = C_V$ is full-dimensional in $\R^n$, the null space
conditions \eqref{eq:defGeneralSettingNSPa}
and \eqref{eq:defGeneralSettingNSPb} simplify a bit: the requirement $v \in
\mathcal{N}(A) \cap (\mathcal{C} + (-\mathcal{C}))$ can be replaced by $v \in
\mathcal{N}(A)$. Moreover, the decomposition $v = v^{(1)} - v^{(2)}$ with
$v^{(1)}$, $v^{(2)} \in \mathcal{C}$ always exists.
\end{remark}

\smallskip

\begin{remark}
  Our setting also captures the constraint that $x$ is known to be
  box-constrained and $x \in \Z^n$. For $\ell,\, u \in \Z^n$ define
  $[\ell, u]_\Z \define \{x \in \Z^n \suchthat \ell \leq x \leq u\}$ to model lower- and upper
  bound constraints of $x \in \Z^n$, where $\le$ is meant component-wise.
  Let $\mathcal{X} = \mathcal{E} = \R^n$, $\mathcal{C} = [\ell, u]_\Z$
  with~$\ell \leq 0 \leq u$. Let $B$
  be the identity map, so that $\mathcal{D} = \mathcal{C}$. Furthermore,
  let $\mathcal{P}$ be the set of orthogonal projectors onto all coordinate
  subspaces of $\R^n$, and define $\overline{P} = I_n - P$, where $I_n$
  denotes the identity mapping on $\R^n$. Define the nonnegative weight
  $\nu(P) \define \rank(P)$, so that $\nu(P)$ is the number of nonzero
  components of the subspace on which $P$ projects. This yields the
  recovery problem
  \begin{align}\label{eq:IntegralBoxConstrained}
    \min\, \{ \norm{x}_1  \suchthat  Ax = b, \; x \in [\ell, u]_\Z\},
  \end{align}
  which has been considered in, e.g.,~\cite{KeiperKLP2017,LangePST2016}.
  Due to the box-constraints $\ell \leq x \leq u$, the last part of
  Assumption~\ref{eq:GeneralSettingAssump1}, $c_1 + c_2 \in \mathcal{C}$
  for all $c_1,\, c_2 \in \mathcal{C}$, is no longer satisfied in general. However,
  an inspection of the proof of Theorem~\ref{thm:GeneralSettingNSP}
  reveals that the condition
  \begin{align*}
    PBc_1,\, \overline{P}Bc_2 \in \mathcal{D} \; \implies \; PBc_1 +
    \overline{P}Bc_2 \in \mathcal{D}
  \end{align*}
  for all $c_1$, $c_2 \in \mathcal{C} + (-\mathcal{C})$ suffices. The
  remaining parts of
  Assumptions~\ref{eq:GeneralSettingAssump1}--\ref{eq:GeneralSettingAssump3}
  as well as Assumption~\ref{eq:GeneralSettingAssump4b} can be proven to be
  satisfied, whereas the counterexample from
  Example~\ref{ex:LinearSettingNSP} shows that
  Assumption~\ref{eq:GeneralSettingAssump4a} is violated in general. Thus,
  by the second part of Theorem~\ref{thm:GeneralSettingNSP}, the null space
  property~\eqref{eq:defGeneralSettingNSPb} characterizes uniform recovery
  for box-constrained integer vectors.
\end{remark}

Since we are especially interested in recovery for (semidefinite)
block-diagonal systems, the next section discusses this setting in detail. More
specifically, this setting is derived from the general setup described
above, and the specific null space properties for characterizing uniform
recovery obtained from~\eqref{eq:defGeneralSettingNSPa} as well
as~\eqref{eq:defGeneralSettingNSPb} are stated.

\section{Semidefinite block-systems}
\label{sec:BlockPSD}

Define $\mathcal{X} = \sym^n$. As linear sensing map consider
the linear operator $A \colon \sym^n \to \R^m$ given by
\[
A(X) = (\sprod{A_1}{X}, \ldots, \sprod{A_m}{X})\T,
\]
where $A_1, \ldots, A_m \in \sym^n$, $b \in \R^m$, and $X \in \sym^n$. The
corresponding matrix equation is then $A(X) = b$.

\begin{remark}
  Note that in this section, we do not follow the notation of the previous section and denote the image of
  some linear map $F$ as $F(X)$ in order to avoid confusion with matrix products.
\end{remark}

The block-diagonal form can now be defined as follows.
\begin{definition}
  \label{de:block1}
  Let $k \geq 1$ and $B_1, \ldots, B_k \ne \varnothing$ a
  partition of the set $[n]$, i.e., $\bigcup_{i=1}^n B_i = [n]$ with pairwise disjoint blocks
  $B_i$.
   A linear operator $A(X)$ is in \emph{block-diagonal form} with
  blocks $B_1, \ldots, B_k$ if and only if $(A_i)_{s,t} = 0$ for all
  $(s,t) \notin (B_1 \times B_1) \cup \cdots \cup (B_k \times B_k)$ and all
  $ i \in [m]$.
\end{definition}
For an index set $I \subseteq [n]$ and a matrix $X \in \sym^n$, write $X_I$
for the submatrix containing rows and columns of $X$ indexed by $I$, and
write $\sym^{I}$ (and $\sym^{I}_+$) as the space of symmetric (positive
semidefinite) $\card{I} \times \card{I}$ matrices with rows and columns
indexed by the elements of $I$.

Let
$\mathcal{E} = \sym^{B_1} \times \cdots \times
\sym^{B_k}$. We write $X \in \mathcal{E}$ as
\begin{align*}
  X =
  \begin{pmatrix}
    X_{B_1} & & \\
    & \ddots & \\
    & & X_{B_k}
  \end{pmatrix} \text{ with } X_{B_i} \in \sym^{B_i} \text{ for
    all } i \in [k].
\end{align*}
The representation map $B \colon \mathcal{X} \to \mathcal{E}$ takes
$X \in \mathcal{X} = \sym^n$ and generates $(X_{B_1},\dots,X_{B_k})\T$ defined as
$X_{B_i} \define \{ (X_{rs})_{r,\,s \in B_i} \}$, $i \in [k]$; note
that entries outside of the blocks are ignored.

The sparsity-induced projections are defined as
$\mathcal{P} = \{P_I \suchthat I \subseteq [k]\}$, where
$P_I\colon \mathcal{E} \to \mathcal{E}$ is the orthogonal projection onto
the subspace
$\mathcal{E}_I \define \{X \in \mathcal{E} \suchthat X_{B_i} = 0 \; \forall
\, i \notin I\}$. The nonnegative weight of a projection
$P_I \in \mathcal{P}$ is defined as $\nu(P) = \card{I}$, and
$\overline{P} \define P_{[k]\bs I}$. Finally, let the norm
$\norm{\cdot}$ be the mixed $*,1$-norm
\begin{align*}
  \norm{X}_{*,1} \define \sum_{i=1}^k \norm{X_{B_i}}_{*},
\end{align*}
where $\norm{\cdot}_{*}$ is the nuclear norm on $\sym^{B_i}$. An element $X \in \mathcal{X}$
is \emph{$s$-block-sparse}, if and only if there exists an index set
\begin{equation}
  \label{eq:block-sparse1}
  \text{$I \subseteq [k]$ with $\card{I} \leq s$ and $P_I(B(X)) = B(X)$},
\end{equation} 
which
implies that $X_{B_i} = 0$ for all $i \notin I$. Thus, we obtain a
block-sparsity setting for matrices. An important side constraint on the
matrix~$X$ which is to be recovered is given by $X \succeq 0$. In order to
model this side constraint in the general setting from
Section~\ref{sec:GeneralSetting}, let $\mathcal{C} = \sym^n_+$ and thus
$\mathcal{D} = \sym^{B_1}_+ \times \cdots \times \sym^{B_k}_+$. In this
case, the general recovery problem~\eqref{eq:GeneralSettingRecoveryProblem}
simplifies to the following convex optimization problem.
\begin{equation}\label{eq:SDP1}
  \min\, \{ \norm{X}_{*,1}  \suchthat  A(X) = b, \; X \succeq 0 \}.
\end{equation}

Define $\norm{x}_0 \define \card{\supp(x)} = \card{ \{i \in [n] \suchthat x_i \neq 0 \}}$
to be the number of nonzero entries in a vector $x$. Then the number of
nonzero blocks in a block-diagonal matrix $X \in \sym^n$ can be written as
\[
\norm{X}_{*,0} =  \norm{ ( \norm{X _{B_1}}_{*}, \ldots, \norm{ X _{B_k} }_{*} )\T  }_0.
\]
Thus, the problem of finding solutions of $A(X) = b$ with minimal number of nonzero blocks
is
\begin{equation}\label{eq:SDP0}
  \min\, \{ \norm{X}_{*,0}  \suchthat  A(X) = b, \; X \succeq 0 \}.
\end{equation}
Then Problem~\eqref{eq:SDP1} is a convex relaxation of~\eqref{eq:SDP0}.

Now we discuss the question when it is possible to recover a block-sparse
positive semidefinite matrix $X^{(0)}$ with $\norm{X^{(0)}}_{*,0} \leq s$,
$s \leq k$, from $b = \mathcal{A}(X^{(0)})$ using the convex
relaxation~\eqref{eq:SDP1}. The next definition provides a null space property
which will be proved to characterize uniform recovery using~\eqref{eq:SDP1}
in Theorem~\ref{thm:nsp_sdp}.

\begin{definition}\label{def:nsp_sdp}
  A linear operator $A(X)$ in block-diagonal form satisfies the
  \emph{semidefinite block-matrix null space property} of order $s$ if and
  only if 
  \begin{align}\label{eq:def_nsp_psd}
    V_{B_i} \preceq 0 \; \forall \, i \in \overline{S} \quad \implies \quad
    \sum_{i \in S} \ones\T \lambda(V_{B_i}) < \sum_{i \in \overline{S}}
    \norm{V_{B_i}}_{*} \tag{$\text{NSP}_{*,1,\succeq 0}^*$}
  \end{align}
  holds for all $V \in (\mathcal{N}(A) \cap \sym^n) \bs \{ 0 \}$ and all
  $S \subseteq [k]$, $\card{S} \leq s$, where
  $ \overline{S} \define [k] \setminus S$ and $\lambda(V_{B_i})$ is the vector of eigenvalues of $V_{B_i}$.
\end{definition}

\begin{theorem}\label{thm:nsp_sdp}
  Let $A(X)$ be a linear operator in block-diagonal form and $s \geq
  1$. The following statements are equivalent:
  \begin{itemize}[leftmargin=4ex]
  \item[(i)] Every $X^{(0)} \in \sym^n_+$ with
    $ \norm{X^{(0)}}_{*,0} \leq s$ is the unique solution of~\eqref{eq:P1}
    with $b = A(X^{(0)})$.
  \item[(ii)] $A(X)$ satisfies the semidefinite block-matrix null space
    property of order $s$.
  \end{itemize}
\end{theorem}

\begin{proof}
  In the situation described above, using $\mathcal{C} = \sym^n_+$,
  $\mathcal{D} = \sym^{B_1}_+ \times \cdots \times \sym^{B_k}_+$ and the
  mixed~$\ell_{*,1}$-norm it is easy to see that
  Assumptions~\ref{eq:GeneralSettingAssump1}--\ref{eq:GeneralSettingAssump3}
  are satisfied. In order to see that
  Assumption~\ref{eq:GeneralSettingAssump4a} holds, let
  $P \define P_S \in \mathcal{P}_s$ be a projection, and consider $V = X - Z$
  with $Z$, $X \in \sym^n_+$ and $P(B(X)) = B(X)$. Let $V^{(1)}$,
  $V^{(2)} \in \sym^n_+$ be a decomposition $V = V^{(1)} - V^{(2)}$. This
  yields
  \begin{align*}
    &\norm{P(B(V^{(1)}))}_{*,1} - \norm{P(B(V^{(2)}))}_{*,1}\\
    & = \sum_{i=1}^n\lambda_i(P(B(V^{(1)}))) - \sum_{i=1}^n\lambda_i(P(B(V^{(2)}))) = \sum_{i=1}^n \lambda_i(P(B(V))) \\
    & = \sum_{i=1}^n \lambda_i(P(B(X))) - \sum_{i=1}^n \lambda_i(P(B(Z))) = \norm{P(B(X))}_{*,1} - \norm{P(B(Z))}_{*,1} \\
    & = \norm{B(X)}_{*,1} - \norm{P(B(Z))}_{*,1} +
          \norm{\overline{P}(B(Z))}_{*,1} - \norm{\overline{P}(B(Z))}_{*,1} \\
    & = \norm{B(X)}_{*,1} - \norm{B(Z)}_{*,1} + \norm{\overline{P}(B(Z))}_{*,1},
  \end{align*}
  and consequently,
  \begin{align*}
    \norm{B(X)}_{*,1} 
    = \norm{B(Z)}_{*,1} - \norm{\overline{P}(B(Z))}_{*,1}
      +\norm{P(B(V^{(1)}))}_{*,1} - \norm{P(B(V^{(2)}))}_{*,1},
  \end{align*}
  which shows, in conjunction with
  $\overline{P}(B(Z)) = \overline{P}(B(V) + B(X))
   = \overline{P}(B(V))$, that Assumption~\ref{eq:GeneralSettingAssump4a} is satisfied.

   It remains to show that~\eqref{eq:defGeneralSettingNSPa} is equivalent
   to~\eqref{eq:def_nsp_psd}.  Therefore, let $S \subseteq [k]$,
   $\card{S} \leq s$ and $P = P_S$ be fixed. Since
  \begin{align*}
    \sum_{i=1}^n \lambda_i(B(V)) = \norm{\lambda_i(B(V^+))}_* - \norm{\lambda_i(B(V^-))}_*,
  \end{align*}
  where $V^+$, $V^- \in \sym^n_+$ with $B(V) = B(V^+) - B(V^-)$,
  Condition~\eqref{eq:def_nsp_psd} clearly
  implies~\eqref{eq:defGeneralSettingNSPa} by choosing $V^{(1)} = V^+$ and
  $V^{(2)} = V^-$.

  For the reverse implication, let again $S \subseteq [k]$, $\card{S} \leq
  s$ and $P = P_S$ be fixed and let $V \in (\mathcal{N}(A) \cap \sym^n)$ with $B(V) \neq 0$ and
  $\overline{P}(B(V)) \preceq 0$. Due to~\eqref{eq:defGeneralSettingNSPa},
  there exist $V^{(1)}$, $V^{(2)} \succeq 0$ with $V = V^{(1)} - V^{(2)}$
  and
  $\norm{P(B(V^{(1)}))}_{*,1} - \norm{P(B(V^{(2)}))}_{*,1} -
  \norm{\overline{P}(B(V))}_{*,1} < 0$. This implies
  \begin{align*}
    0 &> \norm{P(B(V^{(1)}))}_{*,1} - \norm{P(B(V^{(2)}))}_{*,1} -
        \norm{\overline{P}(B(V))}_{*,1} \\
      &= \sum_{i=1}^n \lambda_i(P(B(V^{(1)}))) - \sum_{i=1}^n \lambda_i(P(B(V^{(2)})))
        - \sum_{i=1}^n \abs{\lambda_i(\overline{P}(B(V)))} \\
      &= \sum_{i \in S} \ones\T \lambda(V_{B_i}) - \sum_{i \in \overline{S}}
        \norm{V_{B_i}}_{*}, 
  \end{align*}
  which establishes~\eqref{eq:def_nsp_psd} and by
  Theorem~\ref{thm:GeneralSettingNSP} finishes the proof.
\end{proof}

In order to model the situation where the additional side constraint
$X \succeq 0$ is not present, let $\mathcal{C} = \mathcal{X} = \sym^n$ and
$\mathcal{D} = \mathcal{E} = \sym^{B_1} \times \cdots \times
\sym^{B_k}$, while $A$, $B$, $\mathcal{P}$, $\overline{P}$ and
the norm~$\norm{\cdot}$ are defined as above. In this case, the recovery
problems~\eqref{eq:SDP0} and~\eqref{eq:SDP1} become
\begin{eqnarray}\label{eq:P0}
  & & \min\, \{ \norm{X}_{*,0}  \suchthat  A(X) = b,\; X \in \sym^n \} \\
\label{eq:P1}
  & \text{and} & \min\, \{ \norm{X}_{*,1} \suchthat  A(X) = b,\; X \in \sym^n \},
\end{eqnarray}
respectively. Note that this setting can be obtained
by combining the block/group case and the matrix case
in~\cite{JuditskyKarzanNemirovski2014}.

\begin{definition} \label{def:nsp}
  A linear operator $A(X)$ in block-diagonal form satisfies the
  \emph{block-matrix null space property} of order $s$ if and only if 
  \begin{equation}\label{eq:def_nsp}
    \sum_{i \in S} \norm{V_{B_i}}_* <
    \sum_{i \in \overline{S}} \norm{V_{B_i}}_* \tag{$\text{NSP}_{*,1}^*$}
  \end{equation}
  holds for all $V \in (\mathcal{N}(A) \cap \sym^n) \bs \{ 0 \}$ and all
  $S \subseteq [k]$, $\card{S} \leq s$.
\end{definition}

\begin{theorem}\label{thm:nsp_matrix_block}
  Let $A(X)$ be a linear operator in block-diagonal form and $s \geq
  1$. The following statements are equivalent:
  \begin{itemize}[leftmargin=4ex]
  \item[(i)] Every $X^{(0)} \in \sym^n$ with
    $ \norm{X^{(0)}}_{*,0} \leq s$ is the unique solution
    of~\eqref{eq:P1} with $b = A(X^{(0)})$.
  \item[(ii)] $A(X)$ satisfies the block-matrix null space property of
    order $s$.
  \end{itemize}
\end{theorem}

As already stated, this result can be obtained
from~\cite{JuditskyKarzanNemirovski2014} by combining the block and the
matrix case. Alternatively, it can be derived from the second part of
Theorem~\ref{thm:GeneralSettingNSP}, since it can be shown that
Assumption~\ref{eq:GeneralSettingAssump4b} is satisfied. It is then easy to
verify that~\eqref{eq:defGeneralSettingNSPb} and~\eqref{eq:def_nsp} are
equivalent. The following example shows that we cannot apply the first part
of Theorem~\ref{thm:GeneralSettingNSP}, since
Assumption~\ref{eq:GeneralSettingAssump4a} is violated.

\begin{example}\label{ex:MatrixBlockAssump4aViolated}
  Let $k = 2$ and $n = 4$ together with the partition
  $[4] = \{1,2\} \cup \{3,4\}$, that is, blocks $B_1 = \{1,2\}$ and $B_2 = \{3,4\}$ and consider
  \begin{align*}
    Z = \begin{pmatrix*}[r]
      0 & 0 &   &   \\
      0 & 0 &   &   \\
      &   & 0 & 0 \\
      &   & 0 & 0
    \end{pmatrix*},\, X = \begin{pmatrix*}[r]
      -1 & 0 &   &   \\
      0 & 0 &   &   \\
      &   & 0 & 0 \\
      & & 0 & 0
    \end{pmatrix*}, \, V^{(1)} = \begin{pmatrix*}[r]
      0 & 0 &   &   \\
      0 & 0 &   &   \\
      &   & 0 & 0 \\
      & & 0 & 0
    \end{pmatrix*}, \, V^{(2)} = \begin{pmatrix*}[r]
      1 & 0 &   &   \\
      0 & 0 &   &   \\
      &   & 0 & 0 \\
      & & 0 & 0
    \end{pmatrix*},
  \end{align*}
  together with $V = X - Z = V^{(1)} - V^{(2)}$.  Let $P$ be the projection
  onto the first block coordinates, i.e., $P = P_{\{1\}} \in
  \mathcal{P}$. This yields
  \begin{align*}
    & \norm{B(Z)}_{*,1} + \norm{P(B(V^{(1)}))}_{*,1} - \norm{P(B(V^{(2)}))}_{*,1}
      - \norm{\overline{P}(B(V))}_{*,1} \\
    & = 0 + 0 -1 - 0 < 1 = \norm{B(X)}_{*},
  \end{align*}
  which is a contradiction to Assumption~\ref{eq:GeneralSettingAssump4a}.
\end{example}

\begin{remark}\label{rem:GeneralizationBlockMatrix}
  We could also consider $\mathcal{X}=\mathcal{C}=\R^{n_1 \times n_2}$
  and possibly overlapping blocks $B_i \neq \varnothing$ by
  $B_1 \cup \cdots \cup B_k = [n_1] \times [n_2]$ instead of a partition $B_1,\dots,B_k$ of $[n]$.
  Additionally, we could replace the inner nuclear norms by arbitrary
  norms on $\R^{B_i \times B_i}$. Replacing the inner nuclear norms
  by the chosen norms $\norm{\cdot}$, this also fits in our general setting
  described in Section~\ref{sec:GeneralSetting}, such
  that~\eqref{eq:def_nsp} characterizes uniform recovery
  using
  \begin{align*}
    \min\, \Big\{ \sum_{i=1}^k \norm{X_{B_i}}  \suchthat  A(X) = b,\; X
    \in \R^{n_1 \times n_2} \Big\}.
  \end{align*}
  Note that for block-diagonal matrices we denote by ``inner norms'' the
  norms that are applied on each block. This term is also used later in the
  setting of block-sparse vectors. 
\end{remark}

\section{Interrelations between the NSPs and classification}
\label{sec:Interrelations}

Building upon the general framework, 
we analyze and classify the relative strengths of the NSPs, in particular for
recovery of positive semidefinite block-diagonal matrices and other
prominent subclasses, as well as their interrelations. To prepare for
this, we also briefly record the NSPs for block-sparse (nonnegative)
vectors.
 
\subsection{Block-sparse (nonnegative) vectors}
\label{sec:Blocksparse}

Since every block-sparse vector $x$ can be interpreted as a block-diagonal
matrix $X$ where all blocks are also diagonal matrices, and the
entries of $x$ coincide with the eigenvalues of $X$,
Theorem~\ref{thm:nsp_sdp} also yields a characterization for the uniform
recovery of block-sparse nonnegative vectors using
$\ell_{1,1}$-minimization, which to the best of our knowledge is the first
characterization of uniform recovery in this case.

In order to model such block-structured vectors in the general setup of
Section~\ref{sec:GeneralSetting}, let~$\mathcal{X} = \R^n$. The
block-structure is now given by a partition $B_1,\dots,B_k$ of $[n]$, where
each set $B_i$ is nonempty. Let
$\mathcal{E} = \R^{B_1} \times \cdots \times \R^{B_k}$, where for a
(finite) set $I$, we denote by $\R^I$ the space of elements with entries
indexed by the elements of the set $I$. We write $y \in \mathcal{E}$ as
$y = (y[1],\dots, y[k])\T$, where $y[i] \in \R^{B_i}$ for all $i \in [k]$.
In order to model nonnegativity of $x$, let $\mathcal{C} = \R^n_+$ which
yields $\mathcal{D} = \R^{B_1}_+ \times \cdots \times \R^{B_k}_+$. The
representation map $B \colon \mathcal{X} \to \mathcal{E}$ maps
$x \in \mathcal{C}$ to its block-structured
representation~$y[i] = (x_j)_{j \in B_i}$. The sparsity-induced
projections are given by $\mathcal{P} = \{P_I \suchthat I \subseteq [k]\}$,
where $P_I\colon \mathcal{E} \to \mathcal{E}$ is the orthogonal projection
onto the subspace
$\mathcal{E}_I \define \{y \in \mathcal{E} \suchthat y[i]
= 0 \; \forall \, i \notin I\}$. The nonnegative weight of a projection
$P_I \in \mathcal{P}$ is defined as $\nu(P) = \card{I}$, and
$\overline{P} \define P_{[k]\bs I}$. Finally, let the norm $\norm{\cdot}$
be the mixed $\ell_{1,1}$-norm $\norm{x}_{1,1} \define \sum_{i=1}^k
\norm{y[i]}_{1}$, where~$y = Bx \in \mathcal{E}$ is the block-structured
representation of~$x \in \mathcal{X}$. A vector
$x \in \mathcal{X}$ is \emph{$s$-block-sparse}, if and only if there exists an
index set $I \subseteq [k]$ with $\card{I} \leq s$ and $P_IBx = Bx$, which
for $y = Bx$ implies that $y[i] = 0$ for $i \notin I$. Thus, this
represents block-sparsity, and we arrive at the setting of recovery of
block-sparse nonnegative vectors. This setting has already been considered
in, e.g.,~\cite{EldKB10,elhamifar-vidal-2012,Lin2013,sph-2009} without the
additional nonnegativity constraint, and in, e.g.,~\cite{Stojnic2013} for
nonnegative block-sparse vectors.

The general recovery problem~\eqref{eq:GeneralSettingRecoveryProblem}
yields the recovery problem
\begin{align}\label{eq:P1_BlockLinearNonneg}
  \min\, \{\norm{x}_{1,1} \suchthat Ax = b,\; x \in \R^n_+\},
\end{align}
which is a convex relaxation of the exact recovery problem
\begin{align}\label{eq:P0_BlockLinearNonneg}
  \min\, \{\norm{x}_{1,0} \suchthat Ax = b,\; x \in \R^n_+\}.
\end{align}

By choosing $A(X)$ in diagonal form, the NSP for nonnegative block-linear
systems can be obtained as an immediate corollary of Theorem~\ref{thm:nsp_sdp}.
While it is simply a rephrasing of a special case of that theorem in the language 
of block-sparse vectors, it will help to relate the result further below
to the literature of block-linear systems.

\begin{corollary}\label{cor:nsp_nonneg_block}
  Consider a block-linear system
  $ 
   Ax = \left[ A[1] \cdots A[k] \right] x = b,
  $
  where $b \in \R^m$ and $A \in \R^{m \times n}$ consists of $k$ blocks $A[i] \in \R^{m
    \times n_i}$. The following statements are equivalent:
  \begin{itemize}[leftmargin=4ex]
  \item[(i)] Every $x^{(0)} \in \R^n_+$ with $ \norm{x^{(0)}}_{1,0} \leq s$
    is the unique solution of~\eqref{eq:P1_BlockLinearNonneg} with
    $b = Ax^{(0)}$.
  \item[(ii)] $A$ satisfies the \emph{nonnegative block-linear null space
    property of order $s$}, i.e.,
    \begin{align}\label{eq:nsp_block_linear_nonneg}
      v[\overline{S}] \leq 0 \quad\implies\quad
      \sum_{i \in S} \ones\T v[i] < \sum_{i \in \overline{S}} \norm{v[i]}_1
      \tag{$\text{NSP}_{1,1,\geq 0}$}
    \end{align}
    holds for all $v \in \mathcal{N}(A) \bs \{ 0 \}$ and all
    $S \subseteq [k]$, $\card{S} \leq s$, where
    $v[\overline{S}] \define (v[i])_{i\in \overline{S}}$.
  \end{itemize}
\end{corollary}

In order to model the recovery of block-sparse vectors which are not
necessarily nonnegative, let $\mathcal{C} = \mathcal{X} = \R^n$ and
$\mathcal{D} = \mathcal{E}$, while $B$, $\mathcal{P}$, $\overline{P}$ are
defined as above. This time, we choose the the mixed $\ell_{q,1}$-norm
  $
  \norm{y}_{q,1} \define \sum_{i=1}^k \norm{y[i]}_{q},
  $
with $q \geq 1$ on~$\R^{B_i}$. Note that without the additional constraint $x \geq 0$ it is
not necessary to use an inner $\ell_1$-norm for recovery. The exact
recovery problem using a nonconvex $\ell_0$-term is
\begin{align}\label{eq:P0_BlockLinear}
  \min\, \{\norm{x}_{q,0} \suchthat Ax = b,\; x \in \R^n\},
\end{align}
and its convex relaxation reads
\begin{align}\label{eq:P1_BlockLinear}
  \min\, \{\norm{x}_{q,1} \suchthat Ax = b,\; x \in \R^n\}.
\end{align}
Again, we now formulate a null space property, which, by the first part of
Theorem~\ref{thm:GeneralSettingNSP} can be proved to characterize uniform
recovery using~\eqref{eq:P1_BlockLinearNonneg}.

Similar to the previous section, define the \emph{block-linear null space
property of order $s$} as
\begin{align}\label{eq:block_linearNSP}
  \norm{v[S]}_{q,1} < \norm{v[\overline{S}]}_{q,1}, \tag{$\text{NSP}_{q,1}$}
\end{align}
for all $v \in \mathcal{N}(A)\bs \{0\}$ and all $S \subseteq [k]$ with
$\card{S} \leq s$, where again
$v[S] \define (v[i])_{i\in S}$. This null space
property characterizes the recovery for block-linear systems, as will be
shown in the subsequent corollary. If the inner $\ell_q$-norms are given by
the $\ell_2$-norm, this characterization is due to Stojnic et
al.~\cite{sph-2009}, who state as a remark, that
\begin{quote}
  ``it is reasonable to believe that the null-space characterization [...]
  can easily be generalized to the $\ell_p$ optimization''.
\end{quote}

\begin{corollary}\label{cor:nsp_blocklinear_q}
  Let $A = [A[1] \cdots A[k]] \in \R^{m \times n}$ be in block-linear form
  with $k$ blocks, $x = (x[1],\dots,x[k])\T \in \R^n$ and $s \geq 1$. The
  following statements are equivalent:
  \begin{itemize}[leftmargin=4ex]
  \item[(i)] Every $x^{(0)} \in \R^n$ with
    $ \norm{x^{(0)}}_{q,0} \leq s$ is the unique solution
    of~\eqref{eq:P0_BlockLinear} with $b = Ax^{(0)}$.
  \item[(ii)] $A$ satisfies the block-linear null space property of order
    $s$, i.e.,~\eqref{eq:block_linearNSP} holds for all
    $v \in \mathcal{N}(A)\bs \{0\}$ and all $S \subseteq [n]$ with
    $\card{S} \leq s$.
  \end{itemize}
\end{corollary}

As already stated,
Corollary~\ref{cor:nsp_blocklinear_q} directly follows as a special case
from Theorem~\ref{thm:nsp_matrix_block}. 

\begin{remark}
  Similar to Remark~\ref{rem:GeneralizationBlockMatrix},
  we could also consider $\mathcal{X}=\mathcal{C}=\R^{n}$
  and possibly overlapping blocks $B_i \neq \varnothing$ with $B_1 \cup
  \cdots \cup B_k = [n]$  instead of a
  partition $B_1,\dots,B_k$ of~$[n]$.  Additionally, we could replace the
  inner $\ell_q$-norms by arbitrary norms on $\R^{B_i}$.
  Replacing the inner $\ell_q$-norms by norms $\norm{\cdot}$  also fits in our general setting described in
  Section~\ref{sec:GeneralSetting}, such that~\eqref{eq:block_linearNSP} characterizes uniform recovery using
    $
    \min\, \{ \sum_{i=1}^k \norm{x[i]}  \suchthat  Ax = b \}.
    $
\end{remark}

\begin{table}[H]
  \caption{Null space properties for different settings and their references.}
  \label{tab:nsp_overview}
  \centering
  \begin{tabular}[h!]{@{\extracolsep{\fill}}c@{}@{}c@{}@{}c@{}}
    \toprule
    Setting
    & NSP
    & Reference \\
    \midrule
    \begin{minipage}[h]{0.4\linewidth}
      \centering
      Linear case: \\
      $\min\, \{\norm{x}_1 \suchthat Ax = b, \; x\in \R^n\}$
    \end{minipage}%
    & \begin{minipage}[h]{0.45\linewidth}
      \centering
      $\norm{v_S}_1 < \norm{v_{\overline{S}}}_1$ \\
      $\forall \, v\in \mathcal{N}(A) \bs \{0\}, \; S \subseteq [n],\,
      \card{S} \leq s$.
    \end{minipage}%
    & \begin{minipage}[h]{0.15\linewidth}
      \centering
      \cite{CohenDahmenDevore2009,DonohoHuo2001}, Ex.~\ref{ex:LinearSettingNSP}
    \end{minipage}%
      \vspace{0.5cm} \\
    \begin{minipage}[h]{0.4\linewidth}
      \centering
      Nonnegative linear case: \\
      $\min\, \{\norm{x}_1 \suchthat Ax = b,\; x\in \R^n_+\}$
    \end{minipage}%
    & \begin{minipage}[h]{0.45\linewidth}
      \centering
      $v_{\overline{S}} \leq 0 \, \implies \, \sum\limits_{i\in S} v_i <
      \norm{v_{\overline{S}}}_1$ \\
      $\forall \, v \in \mathcal{N}(A)\bs \{0\}, \; S \subseteq [n],\,
      \card{S} \leq s$.
    \end{minipage}%
    & \begin{minipage}[h]{0.15\linewidth}
      \centering
      \cite{kdxh-2011,Zhang2005}, Ex.~\ref{ex:LinearNonnegSettingNSP}
    \end{minipage}%
      \vspace{0.5cm}\\
    \begin{minipage}[h]{0.4\linewidth}
      \centering
      Block-linear case: \\
      $\min\, \{\norm{x}_{q,1} \suchthat Ax = b,\; x\in \R^n\}$
    \end{minipage}%
    & \begin{minipage}[h]{0.45\linewidth}
      \centering
      $\norm{v[S]}_{q,1} < \norm{v[\overline{S}]}_{q,1}$ \\
      $\forall \, v \in \mathcal{N}(A)\bs \{0\}, \; S \subseteq [k],\,
      \card{S} \leq s$.
    \end{minipage}%
    & \begin{minipage}[h]{0.15\linewidth}
      \centering
      \cite{sph-2009}, \\
      Cor.~\ref{cor:nsp_blocklinear_q}
    \end{minipage}%
      \vspace{0.5cm}\\
    \begin{minipage}[h]{0.4\linewidth}
      \centering
      Nonnegative block-linear case: \\
      $\min\, \{\norm{x}_{1,1} \suchthat Ax = b,\; x\in \R^n_+\}$
    \end{minipage}%
    & \begin{minipage}[h]{0.45\linewidth}
      \centering
      $v[\overline{S}] \leq 0 \implies \,\sum\limits_{i \in S} \ones\T v[i]
      <\norm{v[\overline{S}]}_{1,1}$ \\
      $\forall \, v \in \mathcal{N}(A)\bs \{0\}, \; S \subseteq [k],\,
      \card{S} \leq s$.
    \end{minipage}%
    & \begin{minipage}[h]{0.15\linewidth}
      \centering
      Cor.~\ref{cor:nsp_nonneg_block}
    \end{minipage}%
      \vspace{0.5cm} \\
    \begin{minipage}[h]{0.4\linewidth}
      \centering
      Matrix case:
      $\min\, \{\norm{X}_* \suchthat A(X) = b,\; X \in \sym^n\}$
    \end{minipage}%
    & \begin{minipage}[h]{0.45\linewidth}
      \centering
      $\norm{\lambda_S(V)}_1 < \norm{\lambda_{\overline{S}}(V)}_1$ \\ 
      $\forall \, V \in \mathcal{N}(A) \bs \{0\}, \; S \subseteq [n],\,
      \card{S} \leq s$.
    \end{minipage}%
    & \begin{minipage}[h]{0.15\linewidth}
      \centering
      \cite{oymak-hassibi-2010,RechtXuHassibi2008}, Ex.~\ref{ex:MatrixSettingNSP}
    \end{minipage}%
      \vspace{0.5cm}\\
    \begin{minipage}[h]{0.4\linewidth}
      \centering
      Semidefinite matrix case: \\
      $\min\, \{\norm{X}_* \suchthat A(X) = b,\; X \in \psd^n\}$
    \end{minipage}%
    & \begin{minipage}[h]{0.45\linewidth}
      \centering
      $\lambda_{\overline{S}}(V) \leq 0 \implies  \sum\limits_{j \in S}
      \lambda_j(V) < \norm{\lambda_{\overline{S}}(V)}_1$
      \\ $\forall \, V \in \mathcal{N}(A) \bs \{0\}, \; S \subseteq [n],\,
      \card{S} \leq s$.
    \end{minipage}%
    & \begin{minipage}[h]{0.15\linewidth}
      \centering
      \cite{ksx-2014,oymak-hassibi-2010}, Ex.~\ref{ex:MatrixPSDSettingNSPS}
    \end{minipage}%
      \vspace{0.5cm}\\
    \begin{minipage}[h]{0.4\linewidth}
      \centering
      Block-diagonal case: \\
      $\min\, \{\norm{X}_{*,1} \suchthat A(X) = b,\; X\in \sym^n\}$
    \end{minipage}%
    & \begin{minipage}[h]{0.45\linewidth}
      \centering
      $\sum\limits_{i \in S} \norm{V_{B_i}}_{*} < \sum\limits_{i \in
        \overline{S}}\norm{V_{B_i}}_{*}$ \\
      $\forall \, V \in \mathcal{N}(A) \bs \{0\},\; S \subseteq [k],\,
      \card{S} \leq s$.
    \end{minipage}%
    & \begin{minipage}[h]{0.15\linewidth}
      \centering
      Thm.~\ref{thm:nsp_matrix_block}
    \end{minipage}%
      \vspace{0.5cm}\\
    \begin{minipage}[h]{0.4\linewidth}
      \centering
      Semidefinite block-diagonal case: \\
      $\min\, \{\norm{X}_{*,1} \suchthat A(X) = b,\; X\in \psd^n\}$
    \end{minipage}%
    & \begin{minipage}[h]{0.45\linewidth}
      \centering
      $V_{B_i} \preceq 0 \; \forall \, i \in \overline{S}$ \\
      $\implies \, \sum\limits_{i \in S} \ones\T \lambda(V_{B_i}) <
      \sum\limits_{i \in \overline{S}} \norm{V_{B_i}}_{*}$ \\
      $\forall \, V \in \mathcal{N}(A) \bs \{0\}, \; S \subseteq [k],\,
      \card{S} \leq s$.
    \end{minipage}%
    & \begin{minipage}[h]{0.15\linewidth}
      \centering
      Thm.~\ref{thm:nsp_sdp}
    \end{minipage}%
    \\
    \bottomrule
  \end{tabular}
\end{table}

\subsection{Classification of the null space conditions}
\label{sec:Classification}

Table~\ref{tab:nsp_overview} shows the null space properties for many
important settings considered in existing literature. If already known,
the reference is given in the third column, and if not, the
corresponding theorem (resp.\ corollary) within this paper is
stated. Afterwards, we state
important relationships between the NSPs for the eight
settings considered in Table~\ref{tab:nsp_overview}. 

Recall from~Section~\ref{sec:Blocksparse}
that the block-linear and the nonnegative block-linear cases
are special cases of the block-diagonal and the semidefinite block-diagonal
cases. Note however, that the matrix case and the semidefinite matrix cases
are not special cases of the block-diagonal and the semidefinite block-diagonal
cases (but they still fall into the generalized NSP framework in 
Theorem~\ref{thm:GeneralSettingNSP}).

With respect to the null space properties in
Table~\ref{tab:nsp_overview}, we now compare the conditions that need
to hold in the cases with and without the additional constraints of the
vectors being nonnegative or the matrices being positive semidefinite, when
the inner norms used in the respective recovery problems are identical.

Every NSP for a setting where nonnegativity or
positive semidefiniteness is present stems
from~\eqref{eq:defGeneralSettingNSPa} and the NSPs in the other settings
can be derived from~\eqref{eq:defGeneralSettingNSPb}, see
Section~\ref{sec:GeneralSettingIntroduction}. By definition, a linear sensing map which satisfies~\eqref{eq:defGeneralSettingNSPb}
for $\mathcal{C} = \mathcal{C}_1$ also
satisfies~\eqref{eq:defGeneralSettingNSPa} for every
$\mathcal{C} = \mathcal{C}_2 \subseteq \mathcal{C}_1$, but the converse
needs of course not be true. Thus, in the presence of nonnegativity or
positive semidefiniteness, the conditions needed for characterizing uniform
recovery are not stronger than those needed without this prior knowledge, since
$\mathcal{C}_2 \define \R^n_+ \subseteq \R^n \enifed \mathcal{C}_1$
and~$\mathcal{C}_2 \define \sym^n_+ \subseteq \sym^n \enifed
\mathcal{C}_1$. The following example shows that exploiting positive
semidefiniteness indeed yields a weaker condition for uniform recovery,
when using the nuclear norm as inner norm in both cases.

\begin{example}\label{ex:block_nsp_psd_vs_block_nsp}
  Let $A_1, \ldots, A_4$ be the block-diagonal matrices
  \begin{align*}
    & A_1 = \left( \begin{array}{r@{\,}r@{\,}r@{\;}r}
      0 &   &   &   \\
      &-1 &   &   \\
      &   &-1 & 0 \\
      &   & 0 & 2
    \end{array} \right) ,\,
    A_2 = \left( \begin{array}{r@{\,}r@{\,}r@{\;}r}
      1 &   &   &   \\
      &-1 &   &   \\
      &   &-1 & 0 \\
      &   & 0 &-1
    \end{array} \right), \,
    A_3 = \left( \begin{array}{r@{\,}r@{\,}r@{\;\,}r}
      0 &   &   &   \\
      &-1 &   &   \\
      &   & 1 & 0 \\
      &   & 0 & 0
    \end{array} \right), \,
    A_4 = \left( \begin{array}{r@{\,}r@{\,}r@{\;\,}r}
      0 &   &   &   \\
      & 0 &   &   \\
      &   & 0 & 1 \\
      &   & 1 & 0
    \end{array} \right)
  \end{align*}
  with blocks $B_1 = \{1\}$, $B_2 = \{2\}$ and $B_3 = \{3,4\}$,
  and let $b = (-1,0,0,0)\T$. Consider
  \begin{align}
    \min\, \{ \norm{X}_{*,0} \suchthat A(X) = b,\; X \succeq 0\},
  \end{align}
  where
  $A(X) =
  (\sprod{A_1}{X},\sprod{A_2}{X},\sprod{A_3}{X},\sprod{A_4}{X})\T$, cf.~\eqref{eq:P0}.
  In this case, the null space $\mathcal{N}(A) = \{V \, : \, \sprod{A_i}{V} = 0 \text{ for }i \in [4]\}$ consists exactly of the matrices of the form
   \[
     V = \begin{pmatrix}
              3\alpha & & & \\
              & \alpha  & & \\
              & & \alpha & 0 \\
              & & 0 & \alpha
            \end{pmatrix}, \; \alpha \in \R.
  \]
  Since only nonzero matrices in the null space of $A$ are of interest
  for the NSP, $\alpha$ cannot attain the value $0$.  The
  eigenvalues of $V$ are given by
  $\lambda = (3\alpha,\alpha,\alpha, \alpha)\T$. For the semidefinite
  block-matrix null space property of order $s=1$ to hold, the
  following implications stated in Definition~\ref{def:nsp_sdp} need to
  hold for the support sets $S \in \{\varnothing, \{1\},\{2\},\{3\}\}$:
  \begin{alignat*}{4}
    S &= \varnothing\, : \quad (3\alpha,\alpha,\alpha, \alpha)\T &&\leq 0
    \implies  0 &&< 6\card{\alpha}, \\
    S &= \{1\}\, : \quad (\alpha,\alpha, \alpha)\T &&\leq 0
    \implies  3\alpha &&< 3\card{\alpha}, \\
    S &= \{2\}\, : \quad (3\alpha,\alpha, \alpha)\T &&\leq 0
    \implies  \alpha &&< 5\card{\alpha}, \\
    S &= \{3\}\, : \quad (3\alpha,\alpha)\T &&\leq 0
    \implies  2\alpha &&< 4\card{\alpha}.
  \end{alignat*}
  These are all satisfied, since for every
  $V \in \mathcal{N}(A) \bs \{0\}$, $\alpha \neq 0$ holds.
  
  However, the block-matrix null space property of order $s$ is violated,
  since for $S = \{1\}$, $\alpha \neq 0$
  it holds that
  $\sum_{i\in S}\norm{V_{B_i}}_{*} = 3\abs{\alpha} \geq 3\abs{\alpha} =
  \sum_{i \in \overline{S}}\norm{V_{B_i}}_{*}$, which contradicts~\eqref{eq:def_nsp}.
\end{example}

This demonstrates that explicitly exploiting nonnegativity or positive
semidefiniteness yields stronger results for uniform recovery, which was
already indicated in Remark~\ref{rem:NonnegNSPvsNSP}.
In the next subsection, we strengthen this point by explicitly 
constructing an infinite family of examples that satisfy the nonnegative 
block-linear null space property~\eqref{eq:nsp_block_linear_nonneg}. 
This shows that the proposed null space properties are meaningful in the
sense that they are satisfied by certain general (families of) matrices.

\subsection{An infinite family satisfying the nonnegative block-linear NSP}
\label{sec:ConstructionNonnegBlockNSP}

The NSPs for the nonnegative block-linear case and for
the semidefinite block-diagonal case hold in
many situations. Here,
we provide a specific (rather explicit) infinite family of instances
to show that even for block sizes 
$(n_1, \ldots, n_k) = (2,1, \ldots, 1)$,
the nonnegative block-linear NSP captures cases which are not captured
by the (unrestricted) block-linear NSP and which are not captured 
by the nonnegative linear NSP.

\begin{theorem}
\label{th:constr-family}
Let $k > m \geq 3$ and $B_1, \ldots, B_k$ be blocks of sizes 
$(n_1, \ldots, n_k) \define (2,1, \ldots, 1)$, and set 
$n \define \sum_{i=1}^k n_i = k+1$.
There exists an $m \times n$-matrix 
$A = [A[1]\cdots A[k]]$
such that the nonnegative block-linear NSP 
(see~Corollary~\ref{cor:nsp_nonneg_block})
up to the order
$s^* \define \floor{m/2 - 1}$ is satisfied.
Moreover, for $m \geq 12$
neither the unrestricted block-linear NSP of order $s^*$ is satisfied
nor the nonnegative linear NSP of order $s^*$ is satisfied.
\end{theorem}

In the proof of the theorem, we will the apply the following
characterization of the nonnegative linear NSP.

\begin{proposition}[Donoho \& Tanner~\cite{donoho-tanner-2005}]
\label{pr:outwardly-neighborly}
Let $A \in \R^{m \times n}$ be a matrix with nonzero columns $a^{(1)}, \ldots, a^{(n)}$ 
and $m < n$, and let $s \ge 1$.
Then $A$ satisfies the 
nonnegative linear NSP of order $s$ if and only if 
the polytope 
$P \define \conv \{a^{(1)}, \ldots, a^{(n)}, 0 \}$
has $n+1$ vertices and is outwardly $s$-neighborly, that is,
every subset of $s$ vertices not including the origin span a face of $P$.
\end{proposition}

\begin{remark}\label{re:unrestricted}With 
the same preconditions, $A$ satisfies the unrestricted
linear NSP of order $s$ if and only if the polytope
$P' \define\conv \{\pm a^{(1)}, \ldots, \pm a^{(n)}\}$ has $2n$
vertices and is $s$-centrally neighborly, i.e., any $s$ vertices
not including an antipodal pair span a face of~$P$,
see \cite[Theorem~1]{donoho-2005-neighborly} and also~\cite[Exer.~4.16]{foucart-rauhut-2013}.
By results of
McMullen and Shephard~\cite{mcmullen-shephard-1968}, $P'$ can
never be $s$-centrally neighborly for 
$s > \lfloor (m+1)/3 \rfloor$
(see also~\cite[Section 5.3]{donoho-tanner-techrep-2005}).
\end{remark}

\begin{proof}[Proof of Theorem~\ref{th:constr-family}]
  Let $w^{(1)}, \ldots, w^{(k-1)} \in \R^{m-2} \setminus \{0\}$ be $k-1$ distinct
  points on the moment curve
  $\{(t,t^2, \ldots, t^{m-2})\T \suchthat t \in \R\}$ in $\R^{m-2}$.  It is
  well-known that the polytope $P = \conv \{ w^{(1)}, \ldots, w^{(k-1)}\}$
  is a cyclic polytope, which is $\lfloor (m-2)/2 \rfloor$-neighborly, see,
  e.g.,~\cite[Corollary 0.8]{ziegler-book-1995}. Hence, the nonnegative
  linear NSP of order $\floor{(m-2)/2} = \floor{m/2 - 1}$
  holds for the matrix
  $A' \define [w^{(1)}, \ldots, w^{(k-1)}] \in \R^{(m-2) \times (k-1)}$.

Let $p$ be an interior point of $P$ and 
set
$w' = (p,1,0)\T$,
$w'' = (p,0,1)\T$,
$\hat{w}^{(i)} = (w^{(i)}, 0,0)\T$ for $i \in [k-1]$.
Let $A \define [w',w'',\hat{w}^{(1)}, \dots, \hat{w}^{(k-1)}] \in \R^{m \times n}$
and consider the 
block sizes $(2,1, \ldots, 1)$.
We claim that $A$ satisfies the nonnegative block-linear NSP of order $s^*$. 
Namely, assume that there exists a nonzero vector 
$v = (v_1, \ldots, v_{n})\T \in \mathcal{N}(A) \setminus \{0\}$ 
and $S \subseteq [k]$ with $\card{S} \le s^*$ and $v_{\overline{S}} \le 0$ such that
$\sum_{i\in S} \ones\T v[i] \ge \norm{v_{\overline{S}}}_{1,1}$. Since $v \in \mathcal{N}(A)$
and since the penultimate and the last row of $A$ only have a single
nonzero entry, we have $v_1 = v_2 = 0$. Hence,
$\tilde{v} \define (v_1, v_3, \ldots, v_n)\T$ is a nonzero vector in the
null space 
of $A^{\Diamond} = [w', \hat{w}^{(1)}, \ldots, \hat{w}^{(k-1)}]$ and violates
the nonnegative linear NSP of order $s^*$ for $A^{\Diamond}$.
However, since the 
polytope $P$ and thus also the polytope 
 $\conv \{w', \hat{w}^{(1)}, \ldots, \hat{w}^{(k-1)}\}$ are 
$\floor{m/2 - 1}$-neighborly (due to the pyramidal 
construction with respect to the apex $w'$), this is a contradiction.
  
The nonnegative linear NSP of order $s^*$ does not hold for $A$ if $m \geq 12$, because the
polytope 
$P' \define \conv \{ w',w''$, $\hat{w}^{(1)}, \ldots, \hat{w}^{(k-1)}\}$ is not
$s^*$-neighborly. To see this, observe that any choice 
of vertices which includes $w'$ and $w''$ cannot span a face,
hence $P'$ is not 2-neighborly, and this implies that
$P'$ is not $\floor{m/2 -1}$-neighborly because of $m \ge 6$.

It remains to show that 
the unrestricted block-linear $\text{NSP}_{q,1}$
of order $s^*$ is not satisfied for $m \geq 12$.
Assume that it is satisfied. Then for any 
$v = (v_1, \ldots, v_n)\T \in \mathcal{N}(A) \setminus \{0\}$ and $S \subseteq [k]$
with~$\card{S} \le s^*$, we have
$\norm{v[S]}_{q,1} < \norm{v[\overline{S}]}_{q,1}$
Restricting to $v_1 = 0$, the induced NSP-formula of order $s^*$ must also 
hold for any corresponding $(v_2, \ldots, v_n)\T \in \mathcal{N}(\tilde{A})$, where 
$\tilde{A}$ results from $A$ by deleting the first column,
i.e., $\tilde{A} = [w'', w^{(1)} \ldots, w^{(k-1)}]$.
But this is a contradiction to
the results of McMullen and Shephard from 
Remark~\ref{re:unrestricted},
because we have $m \ge 12$ and thus
$s^* = \floor{m/2 -1} > \floor{(m+1)/3}$.
\end{proof}

\begin{remark}
The construction in the proof can be generalized, for example to
block sizes 
$
  (n_1, \ldots, n_k) = (\underbrace{2, \ldots, 2}_r, \underbrace{1, \ldots, 1}_{n-r})
$
for fixed $r$ and sufficiently large $k$.
\end{remark}

\section{Conclusion and open questions}
\label{sec:Conclusion}

We have presented and discussed a comprehensive framework for 
recovery problems, which, in particular, allows to capture nonnegativity
and positive semidefiniteness constraints. Building upon this framework,
we have established generalized null space conditions, which has also allowed
us to derive a systematic classification of broad classes of NSPs from
a single generalization.

We close the paper by mentioning some open questions.
An important aspect not covered in the present paper concerns the block-structure 
for vectors. Instead of exploiting the knowledge about the block-structure, 
it is possible to directly apply the methods and optimization problems for recovery 
of non-block-structured vectors, by disregarding any information about
blocks. It is clear that in the setting of block-sparse vectors consisting
of~$k$ blocks, every block-$s$-sparse vector~$x$ is also~$\tilde{s}$-sparse
in the classical sense, where~$\tilde{s}$ is the sum of the~$s$ largest
block sizes of the~$k$ blocks, since~$x$ has at most~$\tilde{s}$ nonzero
elements. However, not every sparse vector is also block-sparse with
respect to some block-structure. Thus, the conditions for uniform recovery
of non-block-sparse vectors may be too strong for uniform recovery of all
block-sparse vectors. For a short discussion in terms of the restricted
isometry constant and property, and an illustrative example,
see~\cite{EldarMishali2009}.

In the block-structured settings, an inner $\ell_2$-norm or an inner
Frobenius norm is typically used in the recovery problems due to their
robustness. Without using nonnegativity or positive semidefiniteness, the
respective null space properties can be applied, since these hold for
arbitrary inner norms, see Remark~\ref{rem:GeneralizationBlockMatrix} and
Corollary~\ref{cor:nsp_blocklinear_q}. If nonnegativity or positive
semidefiniteness is exploited in the recovery problem, things seem to be
different. In these cases, the null space properties only hold if the inner
norm is given by the $\ell_1$-norm in the case of vectors or the nuclear
norm in the case of matrices, see Theorem~\ref{thm:nsp_sdp} and
Corollary~\ref{cor:nsp_nonneg_block}. Thus, an interesting line of future
research would be to analyze what happens if another (inner) norm is used
in these block-structured settings.

For non-block settings, however, using different norms often has
side-effects. For example, in the classical case of
sparse recovery, it is well-known that recovery using the $\ell_q$-norm and
the optimization problem
\begin{align}\label{eq:lqMin}
  \min\, \{\norm{x}_q \suchthat Ax = b\}
\end{align}
with $q > 1$ already fails for $1$-sparse vectors, in general, whereas for $0 < q < 1$
the $\ell_q$-norm leads to favorable recovery properties~\cite{MouR10}, but
results in an $\NP$-hard optimization problem~\cite{GeJY11}. Note that
the null space property~\eqref{eq:linearNSP} also
characterizes uniform recovery using~\eqref{eq:lqMin} when replacing the
$\ell_1$-norm by the $\ell_q$-norm with $0 < q < 1$,
see~\cite{foucart-rauhut-2013}.

\section*{Acknowledgement}
We thank the anonymous reviewers for valuable comments and suggestions that
helped to improve the style and presentation of the paper.

\bibliography{block-semidefinite}
\bibliographystyle{abbrv}

\renewcommand{\thefootnote}{\fnsymbol{footnote}}
\setcounter{footnote}{0}

\end{document}